%% file: main.tex
\documentclass[11pt]{article}
\usepackage[letterpaper, margin=1in]{geometry}
\usepackage{amsmath, amssymb, amsthm}
\usepackage{bbm}
\usepackage{graphicx}
\usepackage{enumitem}
\usepackage{url}
\usepackage{hyperref}
\usepackage[capitalise]{cleveref}
\usepackage{algorithm}
\usepackage[noend]{algpseudocode}
\usepackage{subcaption}
\usepackage{xspace}
\usepackage{mathrsfs}
\usepackage{xcolor}
\usepackage{tabularx}
\usepackage{booktabs}
\usepackage{tikz}
\usepackage[sortcites, natbib=true, backref=false, url=false, style=alphabetic, doi=false, isbn=false, minalphanames=3,maxalphanames=4,maxbibnames=9]{biblatex}
\input{preamble}

\hypersetup{
    colorlinks=true,
    linkcolor=black,
    filecolor=blue,      
    urlcolor=blue,
    citecolor=blue,
    pdfpagemode=FullScreen,
}
\DeclareSourcemap{
  \maps[datatype=bibtex]{
    \map{
      \step[fieldset=note, null]
    }
  }
}

\newcommand{\REV}{\operatorname{REV}}
\newcommand{\BREV}{\operatorname{BREV}}
\newcommand{\SREV}{\operatorname{SREV}}
\newcommand{\BSPA}{\operatorname{BSPA}}

\newcommand{\E}{\mathbb{E}}
\newcommand{\Ical}{\mathcal I}
\newcommand{\one}{\mathbf{1}}
\newcommand{\Ent}{\operatorname{Ent}}

\title{A Bulow--Klemperer Theorem for Bundling and Improved Approximation Bounds for Simple Auctions}
\author{Yichen Huang\footnote{Supported in part by NSF grant CNS-2107078. Email: \texttt{yichenhuang@g.harvard.edu}.}\\Harvard University}
\date{\today}
\begin{document}
\maketitle
\begin{abstract}
\pagenumbering{gobble}
We study simple auctions for one additive bidder whose values for $m$ items are independent. Write $\mathrm{REV}$ for the optimal revenue, $\mathrm{SREV}$ and $\mathrm{BREV}$ for the optimal revenues from separate selling and grand bundling, and $\mathrm{BSPA}_2$ for the grand-bundle second-price revenue with two independent bidders.

For independent regular items, we prove \[\mathrm{REV} \le \mathrm{BSPA}_2 \le 2\,\mathrm{BREV}.\] Thus, one additional bidder lets a reserve-free second-price auction for the grand bundle match the optimal revenue from one bidder, regardless of the number of items. This gives a Bulow–Klemperer theorem for bundling and resolves an open question of Beyhaghi, Cai, Feng, Li, and Weinberg (EC 2025) on whether constant competition suffices for bundling. The factor-$2$ bundling approximation is also tight.

For general independent items, we prove \[\mathrm{REV} \le \mathrm{SREV} + \frac{\pi^2}{6}\mathrm{BREV} \le 2.645 \max\{\mathrm{SREV}, \mathrm{BREV}\}.\] This improves the previous approximation factors of $6$ from Babaioff, Immorlica, Lucier, and Weinberg (FOCS 2014, JACM 2020) and $5.2$ from Ma and Simchi-Levi (AISTATS 2021). For the lower bound, we construct an instance in which pricing three disjoint bundles earns more than $2.11\max\{\mathrm{SREV},\mathrm{BREV}\}$, improving the previous lower bound of $2$ from Rubinstein (ITCS 2016). Together, these results narrow the tight approximation ratio to the range $[2.11, 2.645]$.






\newpage
\pagenumbering{arabic}
\end{abstract}

\section{Introduction}

Auction design is a fundamental problem in economics and computer
science~\cite{Daskalakis2015}. For one item and one bidder, an optimal
mechanism is a posted price~\cite{Myerson1981}. Selling several items
is substantially harder: even with a single additive bidder whose item values
are independent, an optimal mechanism can require a continuum of
lotteries~\cite{DaskalakisDeckelbaumTzamos2013}, and exact optimization is
computationally intractable even when each item has only two possible
values~\cite{DaskalakisDeckelbaumTzamos2014}. These difficulties motivate
the study of \emph{simple versus optimal} in auction theory, which asks
how well simple mechanisms can approximate optimal revenue and has
received extensive attention~\cite{HartNisan2017,LiYao2013,BabaioffEtAl2020,MaSimchiLevi2021}.

Two particularly simple selling rules are to sell each item
separately and to sell the grand bundle. Separate selling
adapts prices to the demand for individual items. Bundling pools their
values and can make the bidder's total willingness to pay more
predictable. Either rule can lose a large fraction of optimal revenue
on its own, but the better of the two always yields a constant fraction of the optimal revenue
when values are independent~\cite{BabaioffEtAl2020}.
This raises a quantitative question: how much revenue must a seller
give up to use these simple prices?

This comparison was initiated by Hart and Nisan~\cite{HartNisan2017}. Li and Yao~\cite{LiYao2013} introduced the core--tail
decomposition, which separates the contribution of ordinary values
from rare high values. Building on this approach, \cite{BabaioffEtAl2020} proved that the better of separate selling
and bundling gives a factor-$6$ approximation for independent items.
Ma and Simchi-Levi~\cite{MaSimchiLevi2021} improved the factor to $5.2$,
leaving a gap to the best lower bound of $2$, due to
Rubinstein~\cite{Rubinstein2016}.\footnote{As we were finalizing the manuscript, we became aware of the recent concurrent work~\cite{CaiEtAl2026,GiannakopoulosHahn2026} on the same topic. See \cref{subsec:related-work} for the comparison regarding concurrent work.}


Another way of measuring the relative strength of a simple auction rule is through \emph{enhanced competition}. This line of work evaluates an auction format through its \emph{competition complexity}: the number of extra bidders it needs to match the original optimal revenue. The seminal work of Bulow and Klemperer~\cite{BulowKlemperer1996} showed that, for one item and identically distributed regular\footnote{Regularity is a common assumption in auction theory. A value distribution is regular if its revenue curve is concave, where the curve gives posted-price revenue as a function of the probability of sale. See~\cref{sec:preliminaries} for the formal definition.} bidder values, a second-price auction with one additional bidder earns at least the optimal auction's revenue with the original population. For separate selling, the guarantee worsens immediately with multiple items. For selling $m$ items to a single additive bidder, optimal separate-item auctions require $\Theta(\log m)$ i.i.d. bidders to match the optimal one-bidder revenue~\cite{BeyhaghiWeinberg2019,FeldmanFriedlerRubinstein2018}, even if the items are regular. Bundling offers a different way to use competition and raises the question of whether bundling can recover the constant competition guarantee in this multidimensional setting.

In this work, we give tight competition complexities and improved approximation guarantees
for separate selling and bundling with a single additive bidder. We show that, perhaps surprisingly, for independent regular items, one additional bidder lets a reserve-free auction for the grand bundle match the optimal multidimensional mechanism. Beyond regularity, we also substantially narrow the approximation gap.

\subsection{Our Results}

For one additive bidder with independent nonnegative item values, write
$D$ for their product distribution, $\REV(D)$ for optimal truthful
revenue, and $\SREV(D)$ and $\BREV(D)$ for the revenues from separate
selling and grand bundling. Let $\BSPA_2(D)$ denote second-price revenue
from the grand bundle with $2$ independent bidders drawn from $D$.
Throughout, the distribution of each item's value $V_j$ has finite monopoly revenue, i.e., $\sup_{p \ge 0} \set{p\Pr[V_j \ge p]} < \infty$.

\paragraph{Bundling under Regularity.}
First, we prove tight competition and approximation guarantees for
bundling under regularity. This standard economic assumption expresses
diminishing marginal revenue as the probability of sale increases.
Under regularity, higher values correspond to higher marginal revenue,
so allocating by marginal revenue is consistent with truthful
bidding~\cite{Myerson1981}. Regularity also underlies the
Bulow--Klemperer comparison between competition and auction
optimization~\cite{BulowKlemperer1996}.

\begin{theorem}[Regular items]
\label{thm:regular-bundling-two}
If every marginal is regular, then $\REV(D)\le\BSPA_2(D)\le2\BREV(D)$. The approximation factor $2$ is tight.
\end{theorem}

For this setting with regular items, \cite{BeyhaghiEtAl2025} established
$\REV(D)\le48\BSPA_2(D)$ and $\REV(D)\le18\BREV(D)$.
They also gave a competition-complexity bound of $m$ for BSPA when
$m\in\{2,3\}$, and asked whether constant or sublogarithmic competition
suffices for an arbitrary number of regular items.
\Cref{thm:regular-bundling-two} answers this question in a very strong sense: one additional
bidder suffices, regardless of the number of items. It also strengthens the guarantee in~\cite{FeldmanFriedlerRubinstein2018}, which uses $O_\varepsilon(1)$ bidders to recover a $(1-\varepsilon)$ fraction of $\REV(D)$.

Recall that the classical Bulow--Klemperer theorem~\cite{BulowKlemperer1996} shows that, for one regular item and any number of i.i.d. bidders, one additional bidder lets a second-price auction earn at least as much revenue as the optimal auction for the original population.
\Cref{thm:regular-bundling-two} gives a complementary guarantee: one additional bidder suffices for any number $m$ of independent regular items with one original bidder.\footnote{This does not follow by applying Bulow--Klemperer to the grand bundle. In general, a sum of regular values need not be regular  (see, e.g., an example \href{https://yangpliu.github.io/repo/regular-distributions-not-closed-under-convolution/paper.pdf}{here}); even when it is, the theorem compares $\BSPA_2(D)$ with $\BREV(D)$, rather than with $\REV(D)$.} As in the classical theorem, the auction is prior-independent: $\BSPA_2$ sets no reserve and requires no knowledge of $D$. Thus, one additional bidder is worth at least as much as knowing $D$ and optimizing over all multidimensional mechanisms.


The proof also quantifies the loss from irregularity. For item $j$'s distribution $F_j$, let $\Ical(F_j)$ be the area between item $j$'s revenue curve and its least concave majorant. Recall that a distribution is regular if its revenue curve is concave. Thus, $\Ical(F_j)$ can serve as a definition of \emph{distance to regularity}, which may be of independent interest. For general distributions, we show $\REV(D)\le\BSPA_2(D)+2\sum_j\Ical(F_j)$. As a corollary, this gives a simple proof of $\REV(D)\le2\BREV(D)+\SREV(D)$ for arbitrary independent items (\cref{cor:simple-three}), already improving the previous approximation factor of $5.2$.

\paragraph{Approximation Using $\SREV$ and $\BREV$.}
Working directly with general distributions, we then refine the argument to obtain an approximation factor of $1+\pi^2/6<2.645$ for arbitrary independent item values.
This improves the factor $5.2$ from~\cite{MaSimchiLevi2021} and the
concurrent bounds of $3.52$ and $3$ in~\cite{CaiEtAl2026,GiannakopoulosHahn2026}. 


\begin{theorem}[General independent items]
\label{thm:entropy-upper-summary}
For every product distribution $D$,
\[
 \REV(D)\le\SREV(D)+\frac{\pi^2}{6}\BREV(D)
 \le\left(1+\frac{\pi^2}{6}\right)
       \max\{\SREV(D),\BREV(D)\}.
\]
\end{theorem}

We also improve the lower bound for the first time
since~\cite{Rubinstein2016}, showing that the natural target of
a factor of two is impossible. Together, the bounds place the worst-case
approximation factor in $[2.11,2.645]$.

\begin{theorem}[A lower bound beyond two]
\label{thm:lower-summary}
There is a product distribution $D$ with finitely many items and bounded nonnegative
item values for which posted prices for three disjoint bundles earn
more than $2.11\max\{\SREV(D),\BREV(D)\}$.
\end{theorem}

Finally, we determine how much competition the better of separate selling and bundling needs to recover one bidder's optimal revenue. Beyhaghi and Weinberg~\cite{BeyhaghiWeinberg2019} show that $O(\log m)$ extra bidders suffice for separate selling to recover the optimal revenue, which is tight~\cite{FeldmanFriedlerRubinstein2018}. We show that $\Omega(\log m)$ bidders are necessary even with bundling. Write $\SREV_k(D)$ and $\BREV_k(D)$ for the revenues from optimal separate-item and grand-bundle auctions with $k$ i.i.d. bidders.

\begin{theorem}[Competition for separate selling or bundling]
\label{thm:competition}
There is a universal constant $c>0$ such that, for every sufficiently
large integer $m$, there is a product distribution $D$ on $m$ items such that for every integer $1 \le k \le c \log m$,
\[
 \max\{\SREV_k(D),\BREV_k(D)\}<\REV(D).
\]
\end{theorem}

Interestingly, \cite{FeldmanFriedlerRubinstein2018} show that for any $\varepsilon > 0$, a constant number of extra bidders, depending only on $\varepsilon$, suffices for $\max\{\SREV_k, \BREV_k\}$ to recover a $(1-\varepsilon)$ fraction of the optimal revenue, regardless of $m$. On the other hand, the same paper shows that separate selling requires $\Omega(\log m)$ extra bidders even to approximate. Thus, allowing bundling alongside separate selling removes the dependence on $m$ for every fixed approximation, but falls short for exact recovery.

\subsection{Our Techniques}
\label{subsec:our-technique}
Here, we give an overview of the techniques behind~\cref{thm:regular-bundling-two} and \cref{thm:entropy-upper-summary}. Our proofs take a very different approach from prior and concurrent work~\cite{CaiDevanurWeinberg2016,CaiEtAl2026,LiYao2013,BabaioffEtAl2020,GiannakopoulosHahn2026}, using no duality arguments or core--tail decompositions.

\paragraph{Setting up.}
Let $V=(V_1,\ldots,V_m)$ be the random value vector, with realization $v=(v_1,\ldots,v_m)$, and write $S:=\sum_j V_j$ for the total value. Write $r_j:=\sup_{t\ge0}t\Pr[V_j\ge t]$ for item $j$'s monopoly revenue, so $\SREV=\sum_jr_j$. In this overview, suppose that the item values are bounded. Primes on random variables denote independent copies.

For a scalar random variable $Z$, write $G(Z):=\E[|Z-Z'|]/2$ for its variation. We would like to compare the value that different mechanisms leave uncollected. A second-price auction for the grand bundle with two bidders receives the smaller of two independent bundle values, so its revenue is $\BSPA_2=\E[\min\{S,S'\}]$, and the gap to $\E[S]$ is $G(S)$.

A truthful mechanism instead leaves value uncollected through bidder utility and unallocated items. Fix a truthful menu with nonnegative prices. For realized value $v$, write $x_j(v)\in[0,1]$ for the probability of allocating item $j$, $p(v)$ for the payment, and $u(v):=\sum_j v_jx_j(v)-p(v)$ for the bidder's utility. Let $\ell(v):=\sum_j v_j(1-x_j(v))$ be the unallocated value. Since $S=u(V)+p(V)+\ell(V)$, the menu earns $\E[S]-\E[u(V)]-\E[\ell(V)]$. Thus $\REV\le\BSPA_2$ says that every truthful menu leaves at least $G(S)$ uncollected.

We cannot charge all variation in $S$ to unallocated value. For example, if every realized type purchases the full bundle at a fixed price, all fluctuations in total value become bidder utility, while $\ell(V)=0$. This motivates subtracting utility first and considering the remainder $w(v):=\sum_j v_j-u(v)$. It is the payment plus the value left unallocated, and it is constant in this example. In general,
\[
 w(V)=p(V)+\ell(V),\qquad 0\le w(V)\le S.
\]
It suffices to prove that $G(w(V))\le\E[\ell(V)]$; this is the only step that needs regularity. Once this holds, the identity $w=p+\ell$ gives
\[
 \E[p(V)]=\E[w(V)]-\E[\ell(V)]
 \le\E[w(V)]-G(w(V))
 =\E[\min\{w(V),w(V')\}].
\]
Since $w(V)\le S$, the last expectation is at most $\E[\min\{S,S'\}]=\BSPA_2$.

To bound this variation, we express it through thresholds. For nonnegative $Z$ with cumulative distribution function $F_Z$, the variance of $\one_{\{Z\le t\}}$ is half the probability that two independent draws lie on opposite sides of $t$. Integrating over thresholds gives
\[
 G(Z)=\int_0^\infty\operatorname{Var}(\one_{\{Z\le t\}})\,dt
     =\int_0^\infty F_Z(t)(1-F_Z(t))\,dt.
\]

\paragraph{Controlling the remaining variation.}
Truthfulness connects changes in $w$ to unallocated value. If item $j$'s value increases by $\delta$, the bidder can keep the previously selected option and gain $\delta x_j(v)$ in utility. Therefore, $w$ increases by at most $\delta(1-x_j(v))$, giving $\partial_j^+w(v)\le1-x_j(v)$. Truthfulness also makes $w$ concave, nondecreasing, and $1$-Lipschitz in each coordinate (\cref{lem:menu-residual}). Fixing the other values thus gives a nondecreasing, concave function $g(z):=w(v_{-j},z)$ whose right derivative satisfies $g'_+(z)\le1-x_j(v_{-j},z)$.

For this overview, suppose a regular value $X\sim F$ has a positive density $f$ on its support. Concavity of the revenue curve gives $F(s)(1-F(s))\le sf(s)$ almost everywhere. Applying the threshold representation to $g(X)$ therefore gives
\[
 G(g(X))=\int_0^\infty F(s)(1-F(s))g'_+(s)\,ds
 \le\int_0^\infty sf(s)g'_+(s)\,ds=\E[Xg'_+(X)].
\]
\Cref{lem:regular-scalar} proves this bound for all regular distributions. Applying it with $X=V_j$, we obtain
\[
 G(g(V_j))\le\E_{V_j}[V_jg'_+(V_j)]
 \le\E_{V_j}[V_j(1-x_j(v_{-j},V_j))].
\]
Thus each one-item variation is covered by the corresponding expected unallocated value. Independence then lets us combine these bounds. Comparing two independent value vectors by replacing one coordinate at a time bounds the full variation by the sum of the conditional one-item variations (\cref{clm:coordinate-replacement}). After averaging over the other coordinates, these bounds sum to $\E[\sum_jV_j(1-x_j(V))]=\E[\ell(V)]$. Hence $G(w(V))\le\E[\ell(V)]$, proving $\REV\le\BSPA_2$. The scalar bound $\BSPA_2\le2\BREV$ from \cref{lem:bundle-tail} completes the regularity guarantee in~\cref{thm:regular-bundling-two}.

\paragraph{General distributions without regularity.}
For general marginals, unallocated value alone need not cover variation. We couple a scalar value $X\sim F$ with a larger value $\overline X$ whose revenue curve is the least concave majorant of $F$'s revenue curve. The regular scalar bound applies to $\overline X$. Under this coupling, $\E[\min\{\overline X,\overline X'\}]-\E[\min\{X,X'\}]=2\Ical(F)$, and \cref{lem:irreg-area-control} bounds the shortfall in the scalar inequality by this increase. The same coordinate comparison then gives $G(w(V))\le\E[\ell(V)]+2\sum_j\Ical(F_j)$, and consequently $\REV\le\BSPA_2+2\sum_j\Ical(F_j)$. Since $2\Ical(F_j)\le r_j$ (\cref{lem:area}), this also yields $\REV\le\SREV+2\BREV$.

\paragraph{Strengthening the variation bound.}
Recall that $w(V)=p(V)+\ell(V)$. The preceding coordinate comparison gives $G(w(V))\le\E[\ell(V)]+\SREV$, so
\[
 \E[p(V)]\le\SREV+\E[w(V)]-G(w(V)).
\]
To improve the revenue bound, we seek a stronger lower bound on expected unallocated value. Specifically, we look for a variation statistic $L\ge G$ for which $L(w(V))\le\E[\ell(V)]+\SREV$ still holds.

To this end, recall the threshold representation of $G$. We replace the variance of each threshold indicator by its functional entropy\footnote{Note that this is different from Shannon entropy, although they are indeed closely related. We refer the reader to \cite[Chapter 4]{BoucheronLugosiMassart2013} for more discussions on entropies.}: for bounded $Y\ge0$, define $\Ent(Y):=\E[Y\log Y]-\E[Y]\log\E[Y]$, with $0\log0:=0$ and natural logarithms. For an indicator, the first term vanishes, so $\Ent(\one_{\{Z\le t\}})=-F_Z(t)\log F_Z(t)$. Integrating gives the \emph{cumulative entropy} introduced by Di Crescenzo and Longobardi~\cite{DiCrescenzo2009}:
\[
 L(Z):=-\int_0^\infty F_Z(t)\log F_Z(t)\,dt.
\]
The elementary inequalities $z(1-z)\le-z\log z\le1-z$ for $0\le z\le1$ imply $G(Z)\le L(Z)\le\E[Z]$. Thus $L$ provides the larger deduction we seek. It remains to show that $L(w(V))\le\E[\ell(V)]+\SREV$.

Independence again lets us bound the full variation by adding the one-item variations, with the other values fixed and then averaged. For $L$, this follows from standard entropy tensorization~\cite{Chafai2004} (\cref{lem:entropy-tensorization}).
Fix the other values and write $g(z):=w(v_{-j},z)$. Since $L$ is unchanged by constant shifts and is bounded by the mean, it is enough to bound the mean of the nonnegative increase $g(V_j)-g(0)$. Truthfulness lets us express this increase as the unallocated value of item $j$ plus a payment that could be earned by selling that item alone. The expected payment is at most $r_j$, giving
\[
 L(g(V_j))\le\E_{V_j}[V_j(1-x_j(v_{-j},V_j))]+r_j.
\]
Summing over items gives $L(w(V))\le\E[\ell(V)]+\SREV$. Defining $J(Z):=\E[Z]-L(Z)$, the identity $w=p+\ell$ then gives the comparison in~\cref{thm:entropy-comparison}: $\E[p(V)]\le\SREV+J(w(V))$.

It remains to remove the dependence on the mechanism. We can write $J$ directly in terms of tail probabilities:
\[
 J(Z)=\int_0^\infty\varphi(\Pr[Z>t])\,dt,
 \qquad \varphi(q):=q+(1-q)\log(1-q).
\]
This integral also defines $J$ for unbounded values. Since $\varphi'(q)=-\log(1-q)\ge0$ for $0\le q<1$, $J$ increases with the tail probabilities. Thus $w(V)\le S$ gives $J(w(V))\le J(S)$.

To bound this remaining term by bundle revenue, consider a scalar value $X$ with monopoly revenue $r$. Every price $t>0$ satisfies $t\Pr[X\ge t]\le r$, so its tail is at most $\min\{1,r/t\}$. The equal-revenue distribution attains this tail bound everywhere: every price $t\ge r$ earns $r$. It therefore maximizes $J$ among values with monopoly revenue $r$. Substituting this tail and integrating by parts gives
\[
 J(X)\le r\int_0^1\frac{-\log(1-q)}q\,dq=\frac{\pi^2}{6}r
\]
as shown in \cref{lem:entropy-bundle-cost}. The two-sample minimum instead uses the larger integrand $q^2\ge\varphi(q)$, giving $2r$ on the same tail. This is where entropy improves the coefficient. Applying the bound for $J$ to $S$, whose monopoly revenue is $\BREV$, gives
\[
 \REV\le\SREV+J(S)\le\SREV+\frac{\pi^2}{6}\BREV,
\]
proving~\cref{thm:entropy-upper-summary}.

\subsection{Related Work}
\label{subsec:related-work}
Rochet~\cite{Rochet1987} connects multidimensional incentive constraints
to convexity and cyclic monotonicity. Manelli and
Vincent~\cite{ManelliVincent2006} study when bundling is optimal;
\cite{DaskalakisDeckelbaumTzamos2013,DaskalakisDeckelbaumTzamos2017}
develop optimal-transport and strong-duality approaches to optimal
mechanisms.
Daskalakis~\cite{Daskalakis2015} surveys this literature.

Hart and Nisan~\cite{HartNisan2017} initiated the systematic comparison
of separate selling and bundling with optimal revenue. Li and
Yao~\cite{LiYao2013} introduced the core--tail decomposition and a
constant bundling approximation for i.i.d. items. For general
independent items, \cite{BabaioffEtAl2020} proved an approximation
factor of $6$, improved to $5.2$ by Ma and
Simchi-Levi~\cite{MaSimchiLevi2021}. 

\cite{EdenEtAl2017} initiated multi-item competition
complexity; \cite{BeyhaghiWeinberg2019} and
\cite{DerakhshanEtAl2024} give tight asymptotic
bounds for separate second-price auctions under regularity.
\cite{FeldmanFriedlerRubinstein2018}
combine competition and approximation: for each $\varepsilon>0$, $O_\varepsilon(1)$ i.i.d. bidders suffice to recover $(1-\varepsilon)$ of
one bidder's optimal revenue using separate selling or bundling.
For regular items, their grand-bundle second-price auction alone
suffices. In addition to the results mentioned earlier, \cite{BeyhaghiEtAl2025} show that for MHR (monotone hazard rate) items, a grand-bundle second-price auction with four bidders earns at least the one-bidder welfare $\mathbb{E}[S]$, and that with two or more original bidders, bundling needs a number of additional bidders exponential in~$m$.

\paragraph{Concurrent Work.}
While finalizing this manuscript, after obtaining our main mathematical results,
we became aware of the independent concurrent
papers~\cite{CaiEtAl2026,GiannakopoulosHahn2026}.
\cite{CaiEtAl2026} obtain a factor of
$3.52$ using the duality framework from~\cite{CaiDevanurWeinberg2016} and capped
second moments. Giannakopoulos and Hahn~\cite{GiannakopoulosHahn2026}
obtain a factor of $3$ through mathematical programming and
simple-revenue inequalities. 
Although we also obtain a factor of $3$ as an intermediate result, our techniques differ significantly. Our proof is based on bundle competition under regularity and a
scalar correction for irregularity. The entropy refinement strengthens the variation bound and gives the stronger factor $1+\pi^2/6$.


\section{Preliminaries}
\label{sec:preliminaries}

A seller has $m$ items and faces one additive bidder with value vector
$V:=(V_1,\ldots,V_m)$ drawn from $D:=\prod_{j=1}^m F_j$.
The coordinates are independent and nonnegative. An allocation
$x\in[0,1]^m$ gives the bidder expected value $v\cdot x$.
Each coordinate is the probability of receiving the corresponding item,
so allocations may be fractional when the mechanism randomizes. A
mechanism specifies an allocation $x(v)$ and
an expected payment $p(v)$, giving utility $u(v):=v\cdot x(v)-p(v)$.
Truthfulness requires $u(v)\ge v\cdot x(\widehat v)-p(\widehat v)$ for every
report $\widehat v$, and individual rationality requires $u(v)\ge0$.
Equivalently, the bidder chooses a utility-maximizing option from a menu
containing the outside option $(0,0)$.

We write $\REV(D)$ for the supremum of expected payments over truthful,
individually rational mechanisms, including menus of lotteries. Without loss of generality, we may
restrict attention to nonnegative prices. We take menus to be closed, so that
utility-maximizing choices exist for bounded types.

\paragraph{Simple Auctions and Revenue Benchmarks.}
For a nonnegative scalar value $X\sim F$, the optimal one-bidder revenue
is the best posted-price revenue,
$\REV(F)=\REV(X):=\sup_{t\ge0}t\Pr[X\ge t]$.
Throughout, each marginal has finite revenue $r_j:=\REV(F_j)$.
Selling the items separately earns
$\SREV(D):=\sum_{j=1}^m r_j$.
Write $S:=\sum_{j=1}^m V_j$ for the value of the grand bundle.
Selling this bundle at a single price earns
$\BREV(D):=\sup_{t\ge0}t\Pr[S\ge t]$.

With $k$ bidders, their value vectors are independent draws from $D$.
For a scalar distribution $F$, let $\REV_k(F)$ denote the optimal
single-item auction revenue with $k$ independent draws from $F$.
Selling the items through their optimal single-item auctions earns
$\SREV_k(D):=\sum_j\REV_k(F_j)$.
Auctioning the entire bundle as one object earns
$\BREV_k(D):=\REV_k(S)$, where we also use $\REV_k$ for the distribution
of a scalar random variable. Thus $\SREV_1=\SREV$ and $\BREV_1=\BREV$.

For $k\ge2$, $\BSPA_k(D)$ denotes grand-bundle second-price revenue
without a reserve: the expected second-highest value among $k$ independent
bundle values. The optimal bundle auction defining $\BREV_k$ may use
reserves and
ironing tailored to the distribution.
Competition complexity usually counts the additional bidders needed
to match an original population's optimal revenue. Our subscripts count the total number of bidders; the original benchmark in this
paper has one bidder.

The following standard comparisons relate these benchmarks. The last
inequality is the standard bidder-sampling bound
\cite[Corollary~4.8]{RoughgardenTalgamCohenYan2012}.
We include a short direct proof in Appendix~\ref{app:basic-revenue-bounds} for
completeness.

\begin{lemma}
\label{lem:bundle-tail}
For every product distribution in this model and every integer $k\ge2$,
\[
 \BREV(D)\le m\SREV(D)<\infty,
 \qquad \BSPA_k(D)\le\BREV_k(D)\le k\BREV(D).
\]
\end{lemma}

\paragraph{Regularity.}
For a nonnegative scalar value $X\sim F$, define its upper-tail quantile by
\[
 v_F(q):=\sup\{t\ge0:\Pr[X\ge t]\ge q\},\qquad 0<q<1.
\]
A distribution is \emph{regular} if the revenue curve
$R_F(q):= qv_F(q)$ is concave~\cite{Myerson1981}. For a distribution with positive density $f$ on its
support interval, this is equivalent to monotonicity of its virtual value
$\phi_F(t):=t-(1-F(t))/f(t)$.

We use $\E$ for expectation and $\one_A$ for the indicator of an event $A$. For any random variable $Z$, we use $Z'$ to denote an independent copy.

\section{Competition Complexity of Bundling under Regularity}
\label{sec:regular-bundling-tight}

In this section, we prove~\cref{thm:regular-bundling-two}. We first
establish the comparison with two-bidder bundling and give a tight
two-item example. We then extend the comparison beyond regularity and
interpret the resulting correction through revenue curves.

\subsection{The Regular Comparison}
\label{subsec:regular-bundling-proof}
Recall $G(Z):=\tfrac12\E|Z-Z'|$, half the expected absolute difference
between two independent draws. For an integrable random variable $Z$, it is the
loss from retaining only the smaller draw rather than the mean:
\[
 \E[\min\{Z,Z'\}]=\E[Z]-G(Z).
\]
We will charge this variation to the value a truthful menu leaves
unallocated.

Fix a menu with nonnegative prices, and write $\ell(v):=\sum_jv_j(1-x_j(v))$ for its unallocated value. Following the roadmap in~\cref{subsec:our-technique}, we deduct the bidder's utility from total value and consider $w(v):=\sum_jv_j-u(v)$. This is payment plus
unallocated value:
\[
 w(V)=p(V)+\ell(V),\qquad 0\le w(V)\le S.
\]
We first consider menus with bounded $w$, which we obtain by adding a full-bundle offer at an arbitrarily large price $M$. Taking $M \to \infty$ extends the argument to general menus.

Our key claim is that expected unallocated value covers the variation
of $w(V)$:
\begin{equation}
\label{eq:comparison-reg-goal}
 G(w(V))\le\E[\ell(V)].
\end{equation}
Assuming this bound, let $V'$ be an independent copy of $V$. The identity $w=p+\ell$ gives
\[
 \E[p(V)]\le\E[w(V)]-G(w(V))
 =\E[\min\{w(V),w(V')\}].
\]
Since $w(V)\le S$, the last expectation is at most
$\E[\min\{S,S'\}]=\BSPA_2(D)$. The rest of this subsection proves~\eqref{eq:comparison-reg-goal}.

Truthfulness connects this subgoal to the shape of $w$. When an item's
value increases, the bidder can keep the same menu option, gaining
utility from its allocated fraction. This leaves at most the
unallocated fraction to increase $w$.

\begin{lemma}
\label{lem:menu-residual}
The function $w$ is concave, nondecreasing, and $1$-Lipschitz in each
coordinate. For every value vector $v$ and item $j$, $\partial_j^+w(v)\le1-x_j(v)$.
\end{lemma}
\begin{proof}
Write $(a,b)$ for a menu option with allocation $a\in[0,1]^m$ and
payment $b$. Truthfulness gives $u(v)=\sup_{(a,b)\,\text{in the menu}}\{v\cdot a-b\}$.
As a supremum of affine functions, $u$ is convex, so
$w(v)=\sum_jv_j-u(v)$ is concave. For $\delta\ge0$, increasing $v_j$ by
$\delta$ increases every option's payoff by $\delta a_j\in[0,\delta]$. Taking suprema
gives $0\le u(v+\delta e_j)-u(v)\le \delta$, where $e_j$ is the $j$th unit vector. Consequently,
$0\le w(v+\delta e_j)-w(v)\le \delta$, proving monotonicity and the Lipschitz
bound in each coordinate.

Finally, truthfulness makes the selected allocation $x(v)$ a subgradient
of $u$ at $v$. Thus $\partial_j^+u(v)\ge x_j(v)$, and
\[
 \partial_j^+w(v)=1-\partial_j^+u(v)\le1-x_j(v). \qedhere
\]
\end{proof}

Thus, for a coordinate slice $g(z):=w(v_{-j},z)$, the derivative bound gives
\[
 \E_{V_j}[V_jg'_+(V_j)]
 \le\E_{V_j}[V_j(1-x_j(v_{-j},V_j))].
\]
Here $v_{-j}$ is fixed, and the right side is the expected unallocated value of item $j$.
The next lemma shows that regularity bounds the slice's variation by the left side.\footnote{The overview assumes a positive density for simplicity. The quantile proof here covers general regular distributions, including those with atoms.} Independence will then combine the
scalar bounds.

\begin{lemma}
\label{lem:regular-scalar}
Let $X\sim F$ be nonnegative with a concave revenue curve
$q\mapsto qv_F(q)$, and let
$g:[0,\infty)\to\mathbb R$ be bounded, nondecreasing, concave, and
$1$-Lipschitz. Then
\[
 G(g(X))\le\E[Xg'_+(X)],
\]
where $g'_+ \ge 0$ denotes the right derivative of $g$.
\end{lemma}
\begin{proof}
Write $v(q):=v_F(q)$ and $R(q):=qv(q)$, and let
$R(1):=\lim_{q\uparrow1}R(q)\ge0$. Concavity gives
$R(q)+(1-q)R'(q)\ge R(1)\ge0$ almost everywhere on $(0,1)$.
Substituting $R(q)=qv(q)$ yields
\begin{equation}
 q(1-q)(-v'(q))\le v(q). \label{eq:reg-derivative}
\end{equation}
The smaller of two draws corresponds to the larger of their independent uniform quantiles, which has density $2q$. Thus,
\begin{align*}
     G(g(X))=\int_0^1(1-2q)g(v(q))\,dq &= [q(1-q)g(v(q))]\bigg |_0^1-\int_0^1 q(1-q)(g(v(q)))' dq\\&=
     \int_0^1q(1-q)(-v'(q))g'_+(v(q))\,dq,
\end{align*}
where the second equality follows from integration by parts and the boundary term vanishes by boundedness. Applying the pointwise bound~\eqref{eq:reg-derivative}, with $g'_+\ge0$, proves
\[
 G(g(X))\le\int_0^1v(q)g'_+(v(q))\,dq=\E[Xg'_+(X)]. \qedhere
\]
\end{proof}

The next lemma combines the one-item bounds. We compare two
independent value vectors by replacing one coordinate at a time.
Independence keeps the distribution of each intermediate vector
unchanged, allowing us to bound the full variation by the sum of
conditional scalar variations.

\begin{lemma}
\label{clm:coordinate-replacement}
Let $V=(V_1,\ldots,V_m)$ have independent coordinates, and let
$h:\mathbb R^m\to\mathbb R$ satisfy $\E|h(V)|<\infty$. Then
\[
 G(h(V))\le\sum_{j=1}^m
 \E_{V_{-j}}\bigl[G(h(V_{-j},V_j))\bigr].
\]
Here $V_{-j}$ denotes all coordinates other than $j$, and the $G$ inside
the expectation compares two independent draws of coordinate $j$ while
holding $V_{-j}$ fixed.
\end{lemma}
\begin{proof}
For $0\le j\le m$, set
$V^{(j)}:=(V'_1,\ldots,V'_j,V_{j+1},\ldots,V_m)$.
The triangle inequality gives
\[
 |h(V)-h(V')|\le\sum_j|h(V^{(j-1)})-h(V^{(j)})|.
\]
Take expectations and divide by two. Conditional on the other
coordinates in the $j$th summand, the two values in coordinate $j$ are
independent draws from $F_j$, giving the stated bound.
\end{proof}

Together, these two lemmas show that unallocated value covers the
variation of $w$. We can then compare payments directly with the
minimum of two bundle values.

\begin{proof}[Proof of the upper bound in \cref{thm:regular-bundling-two}]
First consider a menu containing a full-bundle offer at a finite price
$M$. Then $0\le w(V)\le M$. Apply
\cref{clm:coordinate-replacement,lem:regular-scalar} to its coordinate
slices, using $\partial_j^+w\le1-x_j$ from
\cref{lem:menu-residual}, to obtain
\[
 G(w(V))
 \le\sum_j\E_{V_{-j}}[G(w(V_{-j},V_j))]
 \le\E\!\left[\sum_jV_j(1-x_j(V))\right]=\E[\ell(V)].
\]
Since $w(V)=p(V)+\ell(V)$, subtracting this bound gives
\[
 \E[p(V)]\le\E[w(V)]-G(w(V))
 =\E[\min\{w(V),w(V')\}].
\]
The inequality $w(V)\le S$ bounds the last expression by
$\E[\min\{S,S'\}]=\BSPA_2(D)$.

For an arbitrary menu, add a full-bundle offer at price $M$, retaining
the original choice whenever it remains optimal. The new payment
$p_M(v)$ is either $p(v)$ or $M$, so $p_M(v)\ge\min\{p(v),M\}$.
The preceding bound applies to this augmented menu, giving
\[
 \E[\min\{p(V),M\}]\le\BSPA_2(D).
\]
Let $M\to\infty$ by monotone convergence and take the supremum over
menus. Finally, $\BSPA_2(D)\le2\BREV(D)$ follows from
\cref{lem:bundle-tail}.
\end{proof}

\paragraph{A Tight Lower Bound.}
We now show that the approximation $\REV(D)\le2\BREV(D)$ for regular items is tight, proving the lower bound in~\cref{thm:regular-bundling-two}.

\begin{theorem}
\label{thm:regular-bundling-tightness}
For every $\varepsilon>0$, there is a two-item product distribution $D$
with atomless, regular marginals and $\BREV(D)>0$ such that
\[
 \REV(D)\ge(2-\varepsilon)\BREV(D).
\]
\end{theorem}
\begin{proof}
Fix $T\ge1$. Consider independent values $X_T$ and $Y$ defined by
\[
 X_T\sim\mathrm{Unif}\left(\frac{T}{T+1},1\right),
 \qquad \Pr[Y\ge y]:=\frac{1}{1+y}\quad(y\ge0).
\]
Let $D$ be their joint distribution. Both marginals are atomless. The
virtual value of $X_T$ is $2x-1$, and that of $Y$ is $-1$. Both are
nondecreasing, so the marginals are regular.

Posting prices $T/(T+1)$ for $X_T$ and $T$ for $Y$ earns $T/(T+1)$
from each item. Hence $\REV(D)\ge\SREV(D)\ge2T/(T+1)$.
Since $X_T\le1$, every bundle price $p\ge1$ earns at most
\[
 p\Pr[X_T+Y\ge p]\le p\Pr[Y\ge p-1]=1.
\]
Prices below $1$ also earn at most $1$. Thus $0<\BREV(D)\le1$.
The ratio $\REV(D)/\BREV(D)$ is therefore at least $2T/(T+1)$,
which approaches $2$ as $T\to\infty$. Taking $T$ large enough proves
the claim.
\end{proof}

\subsection{Beyond Regularity}
\label{subsec:deficit-implications}
\label{subsec:ironing-interpretation}

Our proof of the upper bound uses regularity only in the scalar comparison
$G(g(X))\le\E[Xg'_+(X)]$. For irregular distributions, we can bound its
shortfall by the area needed to make the revenue curve concave.

Recall the upper-tail quantile $v_F$ from the preliminaries.
The revenue curve of $F$ is $R_F(q):=qv_F(q)$. 
We use the endpoint conventions $R_F(0):=\limsup_{q\downarrow0}R_F(q)$ and
$R_F(1):=\lim_{q\uparrow1}R_F(q)$. The quantile convention retains the
higher value at a downward jump, and $R_F$ is upper
semicontinuous. The regularity condition is exactly concavity of this curve.
Let $\overline R_F$ denote the least concave majorant of $R_F$, also called
its ironed revenue curve. We retain its decreasing portion after the
maximum. Define the area added by Myerson ironing~\cite{Myerson1981} as
\[
 \Ical(F):=\int_0^1(\overline R_F(q)-R_F(q))\,dq.
\]
This area is zero for regular distributions and measures the departure
from concavity otherwise. \Cref{fig:ironing-area} illustrates it. Besides our application here, this natural notion of \emph{distance to regularity} may be of independent interest.

\begin{figure}[tbp]
\centering
\scalebox{0.8}{%
\begin{tikzpicture}[line cap=round,line join=round]
  \begin{scope}[x=1.8cm,y=3.2cm]
    \draw[->] (0,0) -- (3.3,0) node[right] {$x$};
    \draw[->] (0,0) -- (0,1.13);
    \node[below left] at (0,0) {$0$};
    \foreach \x in {1,2,3} {
      \draw (\x,0.015) -- (\x,-0.015);
      \node[below] at (\x,0) {$\x$};
    }
    \draw (-0.035,0.5) -- (0.035,0.5);
    \node[left] at (0,0.5) {$\tfrac12$};
    \draw (-0.035,1) -- (0.035,1);
    \node[left] at (0,1) {$1$};
    \draw[thick] (0,0) -- (1,0.5) -- (2,0.5) -- (3,1) -- (3.13,1);
    \node[above] at (1.5,0.5) {$F(x)$};
  \end{scope}

  \begin{scope}[shift={(10cm,0)},x=5.2cm,y=3.2cm]
    \fill[blue!13]
      (0.5,1) -- (1,0)
      .. controls (0.833333,0.333333) and (0.666667,0.5) .. (0.5,0.5)
      -- cycle;
    \draw[->] (0,0) -- (1.13,0) node[right] {$q$};
    \draw[->] (0,0) -- (0,1.13);
    \node[below left] at (0,0) {$0$};
    \draw (0.5,0.015) -- (0.5,-0.015);
    \node[below] at (0.5,0) {$\tfrac12$};
    \draw (1,0.015) -- (1,-0.015);
    \node[below] at (1,0) {$1$};
    \draw (-0.012,0.5) -- (0.012,0.5);
    \node[left] at (0,0.5) {$\tfrac12$};
    \draw (-0.012,1) -- (0.012,1);
    \node[left] at (0,1) {$1$};

    \draw[thick]
      (0,0) .. controls (0.166667,0.5) and (0.333333,0.833333) .. (0.5,1);
    \draw[thick]
      (0.5,0.5) .. controls (0.666667,0.5) and (0.833333,0.333333) .. (1,0);
    \draw[thick,blue!65!black,dashed] (0.5,1) -- (1,0);
    \fill (0.5,1) circle[radius=1.7pt];
    \draw[thick,fill=white] (0.5,0.5) circle[radius=1.7pt];

    \node[above,text=blue!65!black] at (0.75,0.84) {$\overline R_F(q)$};
    \node[below] at (0.76,0.2) {$R_F(q)$};
    \node[font=\small] at (0.6,0.57) {$\Ical(F)$};
  \end{scope}
\end{tikzpicture}%
}
\caption{An example of $\Ical(F)$. The distribution function of $X \sim \mathrm{Unif}([0,1]\cup[2,3])$ (left), and its revenue curve (right). The shaded region on the right has area $\Ical(F)$.}
\label{fig:ironing-area}
\end{figure}
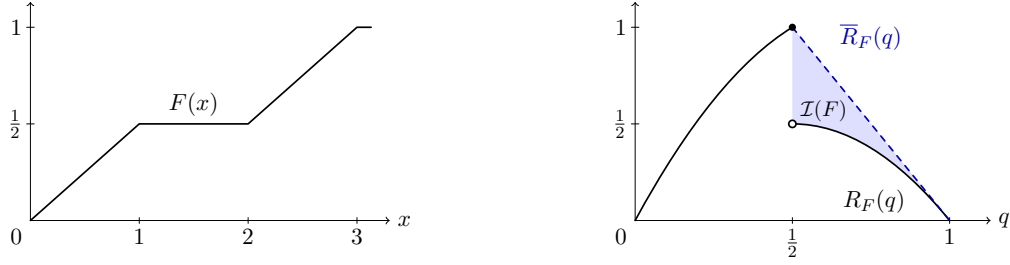

We compare with the distribution obtained by raising values to the
ironed revenue curve. The regular scalar bound applies there, and the
change in the two-sample minimum controls the error. This change is
exactly twice the area added by ironing.

\begin{lemma}
\label{lem:irreg-area-control}
Let $X\sim F$ be nonnegative with finite monopoly revenue, and let
$g:[0,\infty)\to\mathbb R$ be bounded, nondecreasing, concave, and
$1$-Lipschitz. Then $G(g(X))-\E[Xg'_+(X)]\le 2\Ical(F)$.
\end{lemma}
\begin{proof}
By translation, take $g(0)=0$. Put
$p_g(x):=g(x)-xg'_+(x)$. This function is nondecreasing: for
$0\le x\le y$, concavity gives
$g(y)-g(x)\ge(y-x)g'_+(y)$ and $g'_+(x)\ge g'_+(y)$, which imply
$p_g(y)\ge p_g(x)$.
Writing $X_{\min}:=\min\{X,X'\}$, the shortfall becomes
\[
 G(g(X))-\E[Xg'_+(X)]=\E[p_g(X)]-\E[g(X_{\min})].
\]
We will bound it by comparing both terms with those of a value whose
revenue curve is concave.

Define $\overline v(q):=\overline R_F(q)/q$. Concavity of
$\overline R_F$ and its nonnegative left endpoint make $\overline v$
continuous and nonincreasing, so it is the quantile function of a nonnegative value
$\overline X$. Couple $X=v_F(U)$ and $\overline X=\overline v(U)$
using $U\sim\mathrm{Unif}(0,1)$, so $X\le\overline X$. Couple independent copies
in the same way, and put
$\overline X_{\min}:=\min\{\overline X,\overline X'\}$.
The revenue curve of $\overline X$ is concave, so
\cref{lem:regular-scalar} gives
$\E[p_g(\overline X)]\le\E[g(\overline X_{\min})]$.
Monotonicity of $p_g$ and the Lipschitz bound on $g$ therefore give
\[
 \E[p_g(X)]-\E[g(X_{\min})]
 \le\E[g(\overline X_{\min})-g(X_{\min})]\le\E[\overline X_{\min}-X_{\min}].
\]
The larger of two independent uniform quantiles has density $2q$.
Consequently, the increase in the minimum is exactly twice the
ironing area:
\[
 \E[\overline X_{\min}-X_{\min}]
 =2\int_0^1\bigl(\overline R_F(q)-R_F(q)\bigr)\,dq=2\Ical(F). \qedhere
\]
\end{proof}

Replacing the regular scalar bound by \cref{lem:irreg-area-control}
adds the marginal ironing areas to the same coordinate comparison.

\begin{lemma}
\label{lem:regular-bundling-minimum}
For every product distribution $D$, $\REV(D)\le\BSPA_2(D)+2\sum_{j=1}^m\Ical(F_j)$.
\end{lemma}
\begin{proof}
First consider a menu containing a finite-price full-bundle offer, so
$w$ is bounded. \Cref{clm:coordinate-replacement,lem:irreg-area-control}
and $\partial_j^+w\le1-x_j$ give
\[
 G(w(V))\le\E[\ell(V)]+2\sum_j\Ical(F_j).
\]
Consequently,
\[
 \E[p(V)]\le\E[\min\{w(V),w(V')\}]+2\sum_j\Ical(F_j)
 \le\BSPA_2(D)+2\sum_j\Ical(F_j).
\]
The finite-price bundle reduction in
\cref{subsec:regular-bundling-proof} extends the bound to every menu.
Taking the supremum proves the claim.
\end{proof}

The comparison gives a more precise guarantee when the marginals
are close to regular. If $\Ical(F_j)\le\varepsilon r_j$ for every item,
then $\REV(D)\le2\BREV(D)+2\varepsilon\SREV(D)$. Thus the guarantee improves continuously as the area added by ironing decreases. For general distributions, the next lemma shows that $\Ical(F_j) \le r_j/2$.


\begin{lemma}
\label{lem:area}
For every scalar distribution $F$ with finite revenue $r:=\REV(F)$,
$\Ical(F)\le r/2$.
\end{lemma}
\begin{proof}
Write $R:=R_F$, $\overline R:=\overline R_F$, and $v:=v_F$.
Consider a component $(a,b)$ of the open set
$\{q\in(0,1):\overline R(q)>R(q)\}$, allowing $a=0$ or $b=1$.
On this interval, the concave envelope is the chord joining the
endpoint contact points. Upper semicontinuity of $R$ ensures that
these contact points are included.

Since $v$ is nonincreasing, $R(q)\ge qR(b)/b$ for $a<q<b$.
This lower line meets the envelope at $b$, and its gap from the
envelope at $a$ is $R(a)-aR(b)/b$. The area between the two lines is a triangle. Therefore
\begin{align*}
  \int_a^b\bigl(\overline R(q)-R(q)\bigr)\,dq
  &\le\int_a^b\left(\overline R(q)-\frac qbR(b)\right)\,dq=\frac{b-a}{2}\left(R(a)-\frac abR(b)\right)
  \le\frac r2(b-a),
\end{align*}
where the equality computes the triangle's area, and the last inequality uses
$R(a)\le r$ and $R(b)\ge0$.
There are at most countably many such disjoint intervals. Summing
their bounds gives $\mathcal I(F)\le r/2$.
\end{proof}

Consequently, $2 \sum_{j=1}^m \Ical(F_j) \le \SREV(D)$. Together with $\BSPA_2(D)\le2\BREV(D)$, this gives the following factor-three
guarantee, which we improve in the next section.

\begin{corollary}
\label{cor:simple-three}
For every product distribution $D$,
\[
 \REV(D)\le2\BREV(D)+\SREV(D)
 \le3\max\{\SREV(D),\BREV(D)\}.
\]
\end{corollary}

\section{An Entropy Bound for Simple Auctions}
\label{sec:entropy-upper}
We now improve the coefficient $2$ in \cref{cor:simple-three} to
$\pi^2/6$, proving \cref{thm:entropy-upper-summary}. We use a larger
variation term so that subtracting it leaves a smaller bundle
statistic. The new variation must still combine across independent
coordinates, be unchanged by adding a constant, and be at most the
mean of a nonnegative variable.

For bounded nonnegative $X$, the variation used in
\cref{sec:regular-bundling-tight} can also be written as
\[
 G(X)=\int_0^\infty\operatorname{Var}(\one_{\{X\le t\}})\,dt.
\]
We replace variance in this integral with another natural measure of variation, the functional
entropy~\cite{BoucheronLugosiMassart2013,Ledoux1999}.
For a bounded nonnegative random variable $Z$, its functional entropy is
\[
 \Ent(Z):=\E[Z\log Z]-\E[Z]\log\E[Z].
\]
Here $0\log0:=0$, and all logarithms are natural.

\subsection{The Entropy Statistic}
For bounded $X\ge0$ with distribution function $F_X$, define its
cumulative entropy by
\[
 L(X):=\int_0^\infty\Ent(\one_{\{X\le t\}})\,dt
       =-\int_0^\infty F_X(t)\log F_X(t)\,dt.
\]
This quantity was introduced in~\cite{DiCrescenzo2009}; to the best of
our knowledge, its application to mechanism comparison is new.
Like $G$, it vanishes for a deterministic value. A change of variables in the integral shows that it is unchanged by
adding a constant. The inequalities
$z(1-z)\le-z\log z\le1-z$ for $0\le z\le1$ give $ G(X)\le L(X)\le\E[X]$. Thus subtracting $L$ from the mean leaves a nonnegative quantity no larger than the expected two-sample minimum, suggesting a route to improve the factor-three bound in~\cref{cor:simple-three}.

For the bounded residual $w$ from \cref{subsec:regular-bundling-proof},
we will bound $L(w(V))$ by $\E[\ell(V)]+\SREV(D)$.
Subtracting this variation from $\E[w(V)]$ will give the improved
comparison. We first establish a coordinatewise bound for $L$, analogous to the one in \cref{clm:coordinate-replacement}. We use the following standard tensorization of functional entropy, specialized to bounded nonnegative functions. A proof is also included in Appendix~\ref{app:entropy-tensorization} for completeness.

\begin{lemma}[{\cite[Proposition~12]{Chafai2004}}]
\label{lem:functional-entropy-tensorization}
Let $V:=(V_1,\ldots,V_m)$ have independent coordinates, and let $Z$ be a bounded nonnegative measurable function of $V$. Then
\[
 \Ent(Z)\le\sum_{j=1}^m
 \E_{V_{-j}}[\Ent_{V_j}(Z)].
\]
Here $\Ent_{V_j}$ is computed over $V_j$ with $V_{-j}$ fixed;
the outer expectation averages over $V_{-j}$.
\end{lemma}

Applying this inequality to threshold indicators and integrating gives the desired bound for $L$.

\begin{lemma}
\label{lem:entropy-tensorization}
Let $V:=(V_1,\ldots,V_m)$ have independent coordinates, and suppose
$h(V)$ is bounded and nonnegative. Then
\[
 L(h(V))\le\sum_{j=1}^m
 \E_{V_{-j}}[L(h(V_{-j},V_j))].
\]
In each summand, $L$ is computed over $V_j$ with $V_{-j}$ fixed;
the outer expectation averages over $V_{-j}$.
\end{lemma}
\begin{proof}
For each $t\ge0$, apply \cref{lem:functional-entropy-tensorization} to
$Z:=\one_{\{h(V)\le t\}}$ and integrate over $t\ge0$. The left side becomes $L(h(V))$. On the right, Tonelli's theorem interchanges the nonnegative integrals and expectations, giving
\[
 L(h(V))\le\sum_{j=1}^m
 \E_{V_{-j}}\!\left[\int_0^\infty
 \Ent_{V_j}(\one_{\{h(V)\le t\}})\,dt\right]
 =\sum_{j=1}^m\E_{V_{-j}}[L(h(V_{-j},V_j))]. \qedhere
\]
\end{proof}

\subsection{The Mechanism Comparison and Its Implications}
The quantity left after subtracting cumulative entropy will be our bundle
statistic. For bounded nonnegative $X$, define
\[
 J(X):=\E[X]-L(X)\le\E[\min\{X,X'\}].
\]
To use this remainder for unbounded values, we express it directly
through the tail. Put $\varphi(q):=q+(1-q)\log(1-q)$ for
$0\le q\le1$, with $\varphi(1):=1$. Writing the mean and $L$ as tail
integrals gives
\[
 J(X)=\int_0^\infty\varphi(\Pr[X>t])\,dt.
\]
For unbounded $X$, define $J$ by this integral. Since $\varphi$ is
nonnegative and increasing, $J$ is nonnegative and monotone under
stochastic dominance. This lets us compare $J(w(V))$ with $J(S)$ even
when total value is unbounded.

We apply tensorization to $w$. After fixing the other values,
subtracting utility at value zero gives a truthful one-item mechanism
whose revenue is at most $r_j$. Along that coordinate, the increase
in $w$ equals unallocated value plus this one-item payment.
Translation invariance and $L(X)\le\E[X]$ bound the conditional
variation by these two terms. Summing over items then leaves $J(S)$
as the bundle term.

\begin{theorem}
\label{thm:entropy-comparison}
For every product distribution $D$, $\REV(D)\le J(S)+\SREV(D)$.
\end{theorem}
\begin{proof}
By the reduction in \cref{subsec:regular-bundling-proof}, it suffices
to consider menus with nonnegative prices containing a grand-bundle
option at a finite price $M$. Fix any such menu. Then $0\le w(V)\le M$
and $w(V)=p(V)+\ell(V)$.

Fix $j$ and $v_{-j}$. Normalize utility at value zero and define the
corresponding one-item payment by
\[
 u_j(z):=u(v_{-j},z)-u(v_{-j},0),
 \qquad p_j(z):=zx_j(v_{-j},z)-u_j(z).
\]
Restricting reports to this coordinate preserves truthfulness.
Monotonicity of $u$ gives $u_j(z)\ge0$, and truthfulness for the zero
type gives $p_j(z)\ge0$. Thus the induced one-item mechanism is
truthful and individually rational. Independence keeps $V_j$
distributed as $F_j$, so $\E_{V_j}[p_j(V_j)]\le r_j$.
Moreover,
\[
 w(v_{-j},z)-w(v_{-j},0)
 =z-u_j(z)=z(1-x_j(v_{-j},z))+p_j(z).
\]
This difference is nonnegative. Translation invariance and
$L(X)\le\E[X]$ therefore give
\begin{align*}
    L(w(v_{-j},V_j))&=L(w(v_{-j},V_j) - w(v_{-j}, 0))\\&\le \E_{V_j}[V_j(1-x_j(v_{-j},V_j))+p_j(V_j)]
 \\&\le\E_{V_j}[V_j(1-x_j(v_{-j},V_j))]+r_j.
\end{align*}

Applying \cref{lem:entropy-tensorization} and the conditional bound above yields
\begin{align*}
 L(w(V))
 \le\sum_{j=1}^m\E_{V_{-j}}[L(w(V_{-j},V_j))]&\le\sum_{j=1}^m\E_{V_{-j}}\!\left[\E_{V_j}[V_j(1-x_j(V))]+r_j\right]\\
 &=\E\!\left[\sum_{j=1}^mV_j(1-x_j(V))\right]+\sum_{j=1}^m r_j
 =\E[\ell(V)]+\SREV(D).
\end{align*}
Using $w(V)=p(V)+\ell(V)$, we therefore obtain
\[
 \E[p(V)]=\E[w(V)]-\E[\ell(V)]
 \le J(w(V))+\SREV(D).
\]
Since $w(V)\le S$ and $J$ is monotone,
$J(w(V))\le J(S)$. Taking the supremum over menus proves the claim.
\end{proof}

It remains to bound $J(S)$ by bundle revenue. A scalar value with
monopoly revenue $r$ has tail at most $\min\{1,r/t\}$, the
equal-revenue tail. Since $\varphi$ is increasing, this tail maximizes
$J$ among values with monopoly revenue at most $r$.

\begin{lemma}
\label{lem:entropy-bundle-cost}
For every nonnegative scalar $X$ with finite monopoly revenue $r$, $J(X) \le (\pi^2 /6) r$.
\end{lemma}
\begin{proof}
If $r=0$, then $X=0$ almost surely. For $r>0$, bound
$\Pr[X>t]$ by $1$ for $t\le r$ and by $r/t$ for $t>r$.
Monotonicity of $\varphi$ and the substitution $q:=r/t$ give
\[
 J(X)\le r\left(1+\int_0^1\frac{\varphi(q)}{q^2}\,dq\right).
\]
Integration by parts, using $\varphi'(q)=-\log(1-q)$ and
$\varphi(q)=O(q^2)$ at zero, then gives
\[
 1+\int_0^1\frac{\varphi(q)}{q^2}\,dq
 =\int_0^1\frac{-\log(1-q)}q\,dq
 =\sum_{k\ge1}\frac1{k^2}=\frac{\pi^2}{6}.
\]
The second equality follows by expanding the logarithm and applying
Tonelli's theorem.
\end{proof}

Apply \cref{lem:entropy-bundle-cost} to $S$, whose monopoly revenue is
$\BREV(D)$, and combine the result with \cref{thm:entropy-comparison} to improve
the bundle coefficient in \cref{cor:simple-three}.

\begin{corollary}
\label{cor:entropy-general}
For every product distribution $D$,
\[
 \REV(D)\le\SREV(D)+\frac{\pi^2}{6}\BREV(D).
\]
\end{corollary}

This proves \cref{thm:entropy-upper-summary}. For i.i.d. item values,
Kupfer~\cite{Kupfer2016} proves $\SREV(D)\le1.787\BREV(D)$.
Combining this bound with \cref{cor:entropy-general} gives the following
corollary.

\begin{corollary}
\label{thm:iid-bundling}
For i.i.d. item values, $ \REV(D)\le3.432\BREV(D)$.
\end{corollary}


\section{Improved Lower Bounds}
\label{sec:improved-lower}
We now turn to lower bounds. We show that the natural target factor of $2$, suggested by the lower bound
of~\cite{Rubinstein2016}, cannot be achieved: three disjoint bundle prices can
outperform both separate selling and grand bundling by a factor greater
than $2.11$.

The lower bound of $2$ in~\cite{Rubinstein2016} combines two groups of items: one whose revenue
comes mainly from bundling, and another whose revenue comes mainly
from separate prices. To exceed two, we add a third group on which a
menu earns more than separate selling. We call this group the
\emph{seed}: its gain over separate selling supplies the improvement
beyond two. The challenge is to retain this gain without making the
grand bundle too profitable. We isolate the requirement in
\cref{sec:req}, then illustrate it with a very simple two-item seed in \cref{sec:two-item}. Finally,
\cref{sec:lower-analytic} chooses the seed through its total-value
distribution and obtains the factor $2.11$ stated in
\cref{thm:lower-summary}.

\subsection{The Lower Bound Framework}
\label{sec:req}
Write $\SREV_{\mathrm{seed}}$ for the seed's separate revenue,
$\REV_{\mathrm{seed}}$ for the revenue of a chosen truthful menu, and
$S_{\mathrm{seed}}$ for its total value. The seed provides a gain
$\REV_{\mathrm{seed}}-\SREV_{\mathrm{seed}}>0$. To preserve it in a
larger instance, we seek a constant $c$ satisfying $c\ge\SREV_{\mathrm{seed}}$ and $\REV(c+S_{\mathrm{seed}})=c$.

This constant will be the common revenue target for both simple
benchmarks. A bundle price of $c$ sells surely. The second condition
ensures that adding the seed's value to this constant does not raise
bundle revenue: higher prices reduce the sale probability enough to offset the
price increase.

An item with constant value $c$ would give the required control over
bundling, but it would also add $c$ to separate revenue. Instead, we
spread the value across many independent items whose sum concentrates
near $c$ and whose separate revenue tends to zero. A rare, expensive
item then supplies separate revenue $c-\SREV_{\mathrm{seed}}$ while barely
affecting the grand bundle. Both simple benchmarks approach $c$.
Selling the concentrated group together, using the seed menu, and
pricing the rare item separately earns approximately
\[
 c+\REV_{\mathrm{seed}}+(c-\SREV_{\mathrm{seed}})
 =2c+\REV_{\mathrm{seed}}-\SREV_{\mathrm{seed}}.
\]
Thus the concentrated group supplies the revenue from bundling, the
rare item fills the gap in separate revenue, and the seed preserves
its gain over separate selling. This gain, divided by $c$, is the
improvement beyond two. The following lemma implements the
construction with finitely many items.

\begin{lemma}[Completing a seed]
\label{lem:lower-completion}
Consider a finite product of bounded nonnegative item values with
total value $S_{\mathrm{seed}}$, separate revenue $\SREV_{\mathrm{seed}}$,
and a truthful menu earning $\REV_{\mathrm{seed}}$.
If $c>0$, $c\ge\SREV_{\mathrm{seed}}$, and
$\REV(c+S_{\mathrm{seed}})=c$, then, for every
$\varepsilon>0$, there is a finite product distribution $D$ with bounded
nonnegative item values and a truthful menu earning more than
\[
 \left(2+\frac{\REV_{\mathrm{seed}}-\SREV_{\mathrm{seed}}}{c}
       -\varepsilon\right)\max\{\SREV(D),\BREV(D)\}.
\]
\end{lemma}
\begin{proof}
First, construct the concentrated group. Let $Z_N$ take integer values
in $\{1,\ldots,N\}$ with tail $\Pr[Z_N\ge j]:=1/j$, and put
$H_N:=\sum_{j=1}^N1/j$. The relations
\[
 \E Z_N=H_N,\qquad
 \E Z_N^2=\sum_{j=1}^N\frac{2j-1}{j}\le2N,\qquad
 \REV(Z_N)=1
\]
show that $Z_N$ has a growing mean but constant monopoly revenue.
Take $N^2$ independent items, each distributed as
$cZ_N/(N^2H_N)$, and denote their sum by $C_N$. Then
\[
 \E C_N=c,\qquad
 \operatorname{Var}(C_N)\le\frac{2c^2}{NH_N^2}.
\]
The group has separate revenue $c/H_N$. For every fixed
$0<\delta<c$, Chebyshev's inequality shows that a bundle price
$c-\delta$ sells with probability tending to one. Taking
$N$ large and then $\delta$ small makes the group's bundle revenue
arbitrarily close to $c$, while its separate revenue tends to zero.

We next check that replacing the constant by this concentrated group
preserves the bound on grand bundling. For nonnegative random
variables $X,E$ and $0<\theta<1$,
a sale at price $p$ requires either $X\ge(1-\theta)p$ or
$E\ge\theta p$. The union bound and Markov's inequality give
\[
 \REV(X+E)\le\frac{\REV(X)}{1-\theta}+\frac{\E[E]}{\theta}.
\]
Optimizing over $\theta$ yields
\[
 \REV(X+E)\le\left(\sqrt{\REV(X)}+\sqrt{\E[E]}\right)^2.
\]
Apply this with $X:=S_{\mathrm{seed}}+c$ and $E:=|C_N-c|$.
Since $S_{\mathrm{seed}}+C_N\le X+E$ and $\E|C_N-c|\to0$,
bundling the seed with the concentrated group earns at most $c+o(1)$.

Finally, add a rare item to bring separate revenue close to $c$.
Put $d:=c-\SREV_{\mathrm{seed}}$. For each fixed $N$, the combined
value $W$ of the seed and concentrated group is bounded by some $W_{\max}$.
Add an independent item worth $M$ with probability $d/M$ and zero
otherwise, where $M>\max\{W_{\max},d\}$. Its separate revenue is $d$.
Prices at most $W_{\max}$ gain at most $dW_{\max}/M$ from this item; prices between
$W_{\max}$ and $M$ earn at most $d$; and prices above $M$ earn at most
$d(1+W_{\max}/M)$. The new bundle revenue is therefore at most
\[
 \max\{\REV(W),d\}+\frac{dW_{\max}}{M}.
\]

Use the seed menu, sell the concentrated group together, and price
the rare item at $M$. Offers on disjoint groups form a truthful menu.
As $M\to\infty$ and then $N\to\infty$, with the concentrated group's
price approaching $c$, the menu revenue approaches
$2c+\REV_{\mathrm{seed}}-\SREV_{\mathrm{seed}}$, while the larger
simple benchmark approaches $c$. Sufficiently accurate finite choices
prove the claim.
\end{proof}

\subsection{A Two-Item Seed}
\label{sec:two-item}
A two-item example already gives a ratio strictly above $2$.
Let $X$ and $Y$ be independent and identically distributed, each taking
values $0,1,2,3$ with respective probabilities
$3/5,1/5,1/15,2/15$, and put $S_{\mathrm{seed}}:=X+Y$.
Each of the item prices $1,2,3$ earns the same revenue. Compared with
pricing each item at $3$, selling the pair at price $3$ gains sales
at value pairs $(1,2)$, $(2,1)$, and $(2,2)$, while giving a discount at $(3,3)$.
The additional sales more than pay for this discount.

For completion, take $c:=2$. The atom at zero makes higher prices for
$2+S_{\mathrm{seed}}$ unattractive: even a small price increase sharply reduces
the sale probability, and a price of at least $6$ requires both items to be
positive. Three short checks cover all possible prices.

\begin{lemma}
\label{lem:lower-hand-seed}
The two-item seed has separate revenue $\SREV_{\mathrm{seed}}=4/5$.
Selling the two items together at price $3$ earns $\REV_{\mathrm{seed}}=21/25$, and
$\REV(2+S_{\mathrm{seed}})=2$.
\end{lemma}
\begin{proof}
The item tails at $1,2,3$ are $2/5,1/5,2/15$, respectively.
Each price earns $2/5$, and no other price earns more, so
$\SREV_{\mathrm{seed}}=4/5$.
Compared with pricing each item at $3$, the bundle gains a sale on the
three moderate pairs, whose total probability is $7/225$. It loses
$3$ on $(3,3)$, whose probability is $4/225$. The revenue gain is
$3(7-4)/225=1/25$, giving $\REV_{\mathrm{seed}}=21/25$.

The shifted total is integer-valued and at most $8$, so it suffices
to check prices $2,3,\ldots,8$. Price $2$ sells surely. Price $3$
sells unless both items are zero, and earns
\[
 3\bigl(1-(3/5)^2\bigr)=\frac{48}{25}<2.
\]
For prices $4$ and $5$, a sale requires $S_{\mathrm{seed}}\ge2$,
which has probability $1-9/25-6/25=2/5$. Their revenue is at most
$5(2/5)=2$. For prices $6,7,8$, both items must be positive, an event
of probability $(2/5)^2$. Their revenue is at most
$8(2/5)^2=32/25<2$. Hence $\REV(2+S_{\mathrm{seed}})=2$.
\end{proof}

The completion lemma converts the seed's gain of $1/25$ into an
improvement of $(1/25)/2=0.02$ beyond two. Since the seed is sold
at a single bundle price, the completed mechanism posts prices for three
disjoint groups: the seed, the concentrated group, and the rare item.

\begin{theorem}[A lower bound using three bundle prices]
\label{thm:lower-hand-proof}
For every $\varepsilon>0$, there is a finite product distribution $D$
for which posted prices for three disjoint bundles earn more than
\[
 (2.02-\varepsilon)\max\{\SREV(D),\BREV(D)\}.
\]
\end{theorem}

\subsection{An Improved Seed}
\label{sec:lower-analytic}

We obtain the factor $2.11$ by first choosing a total-value
distribution that satisfies the completion condition, then realizing
it with independent items whose separate revenue is small. The total
must often be large enough to sell as a bundle, while higher prices
for the shifted total must reduce the sale probability sufficiently. Concentration alone
cannot achieve this balance: if $S_{\mathrm{seed}}\ge a>0$ surely,
then price $c+a$ sells surely, violating
$\REV(c+S_{\mathrm{seed}})=c$.

Our target combines two independent components. Let $X$ be zero with
probability $3/4$ and otherwise exponentially distributed with rate $1$; let $Y$ be
zero with probability $1/12$ and otherwise exponentially distributed with rate $4$.
The first component is less frequent but larger, while the second is
more frequent but smaller. Put $S_{\mathrm{seed}}:=X+Y$. Conditioning on
which components are present gives
\[
 Q(t):=\Pr[S_{\mathrm{seed}}\ge t]
 =\frac{47}{144}e^{-t}+\frac{11}{18}e^{-4t}
 \qquad(t>0),
\]
and $\Pr[S_{\mathrm{seed}}=0]=1/16$.
The atom at zero makes a small increase above $c$ reduce the sale
probability by a fixed positive amount. The exponential tails control larger price
increases. At the same time, the two components combine often enough
to support a profitable seed bundle.

The components describe the total we want, rather than two physical
items. Putting them in different item populations would let separate
selling choose a price for each value scale. Instead, we express each
component as a Poisson sum of positive contributions and distribute
both kinds of contributions among the same items. Every item then
has the same mixture of the two scales, so separate selling must
choose one price for that mixture. The seed bundle retains the chosen
distribution of total value.

\paragraph{Realizing the Target with Independent Items.}
The compound-Poisson representation uses the nonnegative intensity
\[
 \nu(x):=\frac{e^{-x}-e^{-4x/3}+e^{-4x}-e^{-48x}}{x}
 \qquad(x>0).
\]
Each component contributes one difference of exponentials: an
exponential random variable of rate $a$, present with probability $q\in(0,1)$,
corresponds to the pair of rates $a$ and $a/(1-q)$ in this formula.
Thus our two components give the pairs $(1,4/3)$ and $(4,48)$.
The intensity has total mass
$\lambda:=\int_0^\infty\nu(x)\,dx=\log16$ and first moment
$\int_0^\infty x\nu(x)\,dx=23/48$.

To split this intensity among $n$ independent items, let
$N_i\sim\operatorname{Pois}(\lambda/n)$ be independent counts.
Let $W_{i,n}$ be the sum of $N_i$ positive contributions with density
$\nu/\lambda$, all drawn independently of one another and of the counts. As $n$ grows, each item is
unlikely to receive a contribution, and every contribution is drawn
from the same mixture of the two scales. Their total has exactly the
target law for every $n$: the compound-Poisson identity gives, for
$s\ge0$,
\[
 \E e^{-s\sum_iW_{i,n}}
 =\exp\!\left(-\int_0^\infty(1-e^{-sx})\nu(x)\,dx\right)
 =\frac{(1+3s/4)(1+s/48)}{(1+s)(1+s/4)}.
\]
This is also the Laplace transform of $S_{\mathrm{seed}}$.
The integral identity behind the calculation appears in
Appendix~\ref{app:lower-analytic}.

Splitting the intensity preserves the aggregate law while making
multiple contributions to any one item rare. The proof below shows
that their contribution to separate revenue vanishes uniformly over
prices. We can therefore bound separate revenue through single
contributions while keeping the bundle's total-value distribution
unchanged.

\paragraph{The Revenue Bounds.}
We sell the seed at price $1/2$ and take $c:=1/4$. Write
$\REV_{\mathrm{seed}}:=(1/2)Q(1/2)$ for the revenue at this price.
For separate selling, we first bound the revenue from items receiving
exactly one contribution; the revenue from items receiving multiple
contributions vanishes as $n$ grows. The required bounds come from elementary exponential estimates,
proved in Appendix~\ref{app:lower-analytic}.

\begin{lemma}
\label{lem:lower-analytic-seed}
The prescribed total and intensity satisfy
\[
 \REV(c+S_{\mathrm{seed}})=c=1/4
 \qquad\text{and}\qquad
 \REV_{\mathrm{seed}}>\frac7{50}.
\]
Moreover, for every $t>0$,
\[
 t\int_t^\infty\nu(x)\,dx\le\frac9{80}.
\]
\end{lemma}

To apply the completion lemma, cap each item at the seed price $1/2$.
This makes the items bounded and preserves sales at price
$1/2$: if an item exceeds the cap, its capped value alone still
suffices for a sale. Capping can only decrease separate revenue and
the shifted total, so it also preserves the other completion
requirements. The resulting seed has separate revenue at most $9/80+o(1)$, while its bundle revenue is strictly greater than
$7/50$. For sufficiently many items, its gain is therefore greater than
\[
 \frac7{50}-\frac9{80}=\frac{11}{400}.
\]
Relative to $c=1/4$, this is an improvement of $0.11$ beyond two.

\begin{theorem}
\label{thm:lower-analytic}
There is a finite product distribution $D$ with bounded nonnegative
item values for which posted prices for three disjoint bundles earn
more than
\[
 2.11\max\{\SREV(D),\BREV(D)\}.
\]
\end{theorem}
\begin{proof}
Cap the independent items $W_{i,n}$ at $1/2$, giving
$V_{i,n}:=\min\{W_{i,n},1/2\}$, and denote their total by
$S_{\mathrm{seed},n}$. Capping preserves sales at the seed price and
only decreases the total. Thus the capped seed still earns
$\REV_{\mathrm{seed}}$ at price $1/2$. By
\cref{lem:lower-analytic-seed} and the fact that price $c$ sells surely,
\[
 \REV(c+S_{\mathrm{seed},n})=c.
\]

It remains to bound the capped seed's separate revenue. For a price
$t$, first consider items receiving exactly one contribution:
\[
 nt\Pr[W_{1,n}\ge t,\ N_1=1]
 =e^{-\lambda/n}t\int_t^\infty\nu(x)\,dx
 \le\frac9{80}e^{-\lambda/n}.
\]
For items receiving two or more contributions, revenue is bounded by
their expected value. Uniformly over all prices, we have
\[
 nt\Pr[W_{1,n}\ge t,\ N_1\ge2]
 \le n\E[W_{1,n}\mathbf1_{\{N_1\ge2\}}]
 =\frac{23}{48}(1-e^{-\lambda/n}).
\]
Here the mean contribution is $23/(48\lambda)$, and the Poisson count
satisfies
$\E[N_1\mathbf1_{\{N_1\ge2\}}]
=(\lambda/n)(1-e^{-\lambda/n})$.

These bounds hold for every price, and capping cannot increase monopoly
revenue. Therefore
\[
 \SREV_{\mathrm{seed},n}
 \le n\REV(W_{1,n})
 \le\frac9{80}e^{-\lambda/n}
     +\frac{23}{48}(1-e^{-\lambda/n})
 =\frac9{80}+o(1).
\]
For sufficiently large $n$, we have
$\SREV_{\mathrm{seed},n}<c$ and
$\REV_{\mathrm{seed}}-\SREV_{\mathrm{seed},n}>11/400$.
The completion lemma then gives a ratio arbitrarily close to
\[
 2+\frac{\REV_{\mathrm{seed}}-\SREV_{\mathrm{seed},n}}{c}
 >2+\frac{11/400}{1/4}=2.11.
\]
Choose the completion error smaller than the strict surplus above $2.11$.
The completed mechanism posts one price for the seed, one for the
concentrated group, and one for the rare item, giving three disjoint
bundle prices as claimed.
\end{proof}

\section{Competition with Separate Selling and Bundling}
\label{sec:competition}
We now prove the lower bound in~\cref{thm:competition}.
For regular items, two bidders suffice for bundling to match one
bidder's optimal revenue. For general independent items, even the
better of separate selling and bundling needs $\Omega(\log m)$
bidders for exact recovery. This matches the $O(\log m)$ upper bound
for separate selling alone~\cite{BeyhaghiWeinberg2019}.
By \cref{thm:regular-bundling-two}, such an instance must violate regularity;
the two-point rare items below have nonconcave revenue curves.

The construction must make both selling rules lose revenue.
A sum of many capped equal-revenue items has a large, concentrated
bundle value but little separate revenue. We add a second group of
rare items whose positive values lie on widely separated scales,
following Rubinstein~\cite[Example~1]{Rubinstein2016}. These items
earn substantial revenue separately, while their sum has constant
monopoly revenue. One bidder can be offered the first group as a bundle and the
rare items at separate prices. Separate selling misses the first group's
bundle revenue, while grand bundling misses revenue from the rare
items even with additional bidders.

To bound grand-bundle revenue, we separate these two contributions.
The next lemma bounds the revenue from $X+Y$ by the optimal revenue
from $X$ plus the expected maximum of the $Y$ values. After fixing the $Y$
values, subtracting each bidder's allocated $Y$ value from its payment
leaves an auction for $X$. The total amount subtracted is at most the
largest $Y$ value, because only one object can be allocated.
The first inequality is the marginal-mechanism bound
in~\cite[Lemma~1]{BabaioffEtAl2020}; the second follows from
\cref{lem:bundle-tail}. We include a proof in
Appendix~\ref{app:basic-revenue-bounds} for completeness.

\begin{lemma}[{\cite[Lemma 1]{BabaioffEtAl2020}}]
\label{lem:competition-scalar-components}
Let $X$ and $Y$ be independent nonnegative random variables. For $k$
independent bidders whose values are distributed as $X+Y$,
\[
 \REV_k(X+Y)\le\REV_k(X)+\E\!\left[\max_{i\le k}Y_i\right]
 \le k\REV(X)+\E\!\left[\max_{i\le k}Y_i\right],
\]
where $Y_1,\ldots,Y_k$ are independent copies of $Y$.
\end{lemma}

We choose enough items in the first group that its total value stays
close to its mean for each of $k$ bidders. The lemma then limits
grand-bundle revenue to roughly this mean, plus $k$ times the rare
group's constant monopoly revenue. The rare items' one-bidder separate
revenue covers this contribution and the small losses from
concentration. The construction uses $e^{O(k)}$ items, yielding the
logarithmic bidder lower bound.

\begin{lemma}[A loss that persists with additional competition]
\label{lem:competition-two-block}
For every integer $k\ge64$, there is a product distribution $D$
over at most $e^{20k}$ items such that
\[
 \max\{\SREV_k(D),\BREV_k(D)\}<\REV(D).
\]
\end{lemma}
\begin{proof}
Fix $k\ge64$, and put $T:=e^{16k-1}$ and $n:=\lceil k^2T\rceil$.
Let $Y_T$ be the capped equal-revenue variable supported on $[1,T]$,
with $\Pr[Y_T\ge t]=1/t$ for $1\le t\le T$ and an atom at $T$.
Direct integration gives
\[
 \REV(Y_T)=1,\qquad \E[Y_T]=1+\log T=16k,\qquad
 \operatorname{Var}(Y_T)\le2T.
\]

\paragraph{A Concentrated Group.}
Take $n$ independent copies of $Y_T$, and let $C$ be their sum.
Since $\E[C]=16kn$ and $\operatorname{Var}(C)\le2nT$,
Chebyshev's inequality gives
\[
 \Pr[C<\E[C]-n]\le\frac{2nT}{n^2}\le\frac{2}{k^2}.
\]
The bundle price $\E[C]-n$ therefore earns at least
\[
 (\E[C]-n)(1-2/k^2)\ge \E[C]-n-\frac{32n}{k}\ge \E[C]-\frac32n.
\]
In contrast, separate selling to $k$ bidders earns at most $kn$ from
this group, because each item has one-bidder revenue one.

\paragraph{A Group with Separated Rare Values.}
For $i=1,\ldots,4n$, take an independent item $Z_i$ worth $2^i$ with
probability $2^{-i}$ and zero otherwise. Put $Z:=\sum_iZ_i$.
Each item has monopoly revenue one, so their total separate revenue
is $4n$.

To bound the revenue of their sum, fix a positive price $p$ no larger
than its maximum possible value. Let $j$ be the least index with
$p\le\sum_{i=1}^j2^i$. All lower-index values together fall below $p$,
so a sale requires some active item with index at least $j$. Hence
\[
 \Pr[Z\ge p]\le\sum_{i=j}^{4n}2^{-i}\le2^{1-j}.
\]
Since $p<2^{j+1}$, every such price earns at most four. Higher prices
earn zero, and therefore $\REV(Z)\le4$.

\paragraph{Comparing the Revenues.}
A menu can offer the first group at price $\E[C]-n$ and the rare
items at separate prices. Thus
\[
 \REV(D)\ge \E[C]-\frac32n+4n=\E[C]+\frac52n.
\]
Separate selling to $k$ bidders earns at most $k$ times the one-bidder
separate revenue, so
\[
 \SREV_k(D)\le k(n+4n)=5kn<\E[C].
\]

For grand bundling, apply~\cref{lem:competition-scalar-components}
with $X:=Z$ and $Y:=C$. Since $\REV(Z)\le4$, we have
\[
 \BREV_k(D)\le4k+\E\!\left[\max_{i\le k}C_i\right].
\]
The inequalities $\max_i C_i\le \E[C]+\sum_i(C_i-\E[C])^+$ and
$\E[(C-\E[C])^+]\le\sqrt{\operatorname{Var}(C)}$ give
\[
 \BREV_k(D)\le \E[C]+4k+k\sqrt{2nT}\le \E[C]+4k+\sqrt2n.
\]
Here we used $n\ge k^2T$. Also, $n\ge k^2\ge64k$ gives
$4k\le n/16$. Since $\sqrt2+1/16<3/2$, it follows that
\[
 \max\{\SREV_k(D),\BREV_k(D)\}\le \E[C]+\frac32n
 <\E[C]+\frac52n\le\REV(D).
\]
Finally, the number of items satisfies
\[
 5n\le10k^2e^{16k-1}\le e^{20k}. \qedhere
\]
\end{proof}

For sufficiently large $m$, take $k:=\lfloor\log(m)/20\rfloor\ge64$.
The construction uses at most $e^{20k}\le m$ items. Adding deterministic
zero-valued items gives exactly $m$ items without changing any revenue
benchmark. The better of separate selling and bundling still earns
less than $\REV(D)$ with $k$ bidders. Both simple revenues are nondecreasing in the number of bidders, so the same strict inequality holds with fewer bidders. This proves \cref{thm:competition}.

\section{Conclusion and Discussion}
We obtain tight competition and approximation guarantees for bundling
under regularity. An entropy refinement of the same argument gives a
$(1+\pi^2/6)$-approximation by the better of separate selling and
bundling. The lower-bound construction shows that the general approximation
factor must exceed two. We also show that even the better of separate selling and bundling can require
$\Omega(\log m)$ bidders to match one bidder's optimal revenue. An immediate open problem is to close the
approximation gap. Another direction is to find simple auction formats
that compete with optimal revenue from several original bidders.

\paragraph{AI Methodology.} The results were developed with the assistance of ChatGPT-6 Astra.
The author first observed, without AI assistance, that
the variation term can be useful for proving $\REV(D)\le\BSPA_2(D)$ for regular
item distributions. The AI subsequently generalized the argument to
irregular distributions at the author's request. The author then suggested exploring other forms of variation, and the AI developed the entropy refinement in \cref{sec:entropy-upper}. For \cref{sec:improved-lower}, the author proposed the idea of a third group of items, and the AI formalized the idea and carried out the numerical search. The selected parameters and all claimed bounds are verified analytically. The mathematical ideas in \cref{sec:competition} are the author's, developed with little AI input. ChatGPT-6.1 Sol and Claude Opus 5.5 were used for proofreading the manuscript and improving the exposition. The author takes full responsibility for the originality and correctness of all AI-assisted content.

\paragraph{Acknowledgements.} I thank Yannai A. Gonczarowski for introducing the problem to me and
for many insightful discussions, particularly his question about whether the
correction in \cref{sec:regular-bundling-tight} admits an interpretation
as a distance to regularity. I also thank Sadie Zhao and Michael Mitzenmacher for
many useful comments on earlier versions of the paper.

\appendix

\section{Basic Revenue Bounds}
\label{app:basic-revenue-bounds}

A bundle sale implies that some item has value at least a $1/m$
fraction of the bundle price. For the bidder-sampling bound, placing
one real bidder among simulated bidders recovers a $1/k$ fraction of
the auction's revenue.

\begin{proof}[Proof of~\cref{lem:bundle-tail}]
For every $t>0$, the union bound gives
\[
 t\Pr[S\ge t]
 \le\sum_{j=1}^m t\Pr[V_j\ge t/m]
 \le m\sum_{j=1}^m\REV(F_j)=m\SREV(D).
\]
Taking the supremum over $t$ proves the first comparison.
The no-reserve second-price auction is feasible, so
$\BSPA_k(D)\le\BREV_k(D)$.
For the last inequality, take any $k$-bidder bundle auction. Place one
real bidder in a uniformly chosen position and simulate the other
$k-1$ independent bundle values. Retain only the real bidder's
allocation and payment. This is a truthful, individually rational
one-bidder bundle mechanism whose expected revenue is a $1/k$ fraction of the
original revenue. Thus every such auction earns at most $k\BREV(D)$.
\end{proof}

For the component bound, fixing the $Y$ values and subtracting their
allocated value from payments leaves an auction for $X$. The total
amount subtracted is at most the largest $Y$ value.

\begin{proof}[Proof of~\cref{lem:competition-scalar-components}]
Take a dominant-strategy truthful, ex-post individually rational
scalar auction, and fix the vector of $Y$ values at $y$. If bidder $i$
receives allocation probability $a_i(x+y)$ and pays $p_i(x+y)$, charge
$p_i(x+y)-y_i a_i(x+y)$ instead. Utility is unchanged, so this gives
a truthful, individually rational auction for the $X$ values.
The removed payments sum to at most $\max_i y_i$, because
$\sum_i a_i\le1$. Thus the original auction's expected revenue,
conditional on $y$, is at most $\REV_k(X)+\max_i y_i$. Averaging over
$y$ proves the first inequality. Optimal scalar revenue can be
attained or approached by such auctions~\cite{Myerson1981}.
The second inequality follows from the one-item case of
\cref{lem:bundle-tail}, with the $k=1$ case immediate.
\end{proof}

\section{Entropy Tensorization}
\label{app:entropy-tensorization}

The entropy chain identity separates one coordinate from the rest. Convexity lets us compare the entropy of an average with the average of the entropies; induction then gives tensorization.

\begin{proof}[Proof of \cref{lem:functional-entropy-tensorization}]
Let $V$ and $Z$ be as in the lemma. A subscript on $\E$ or $\Ent$ specifies the coordinates averaged over, with the others held fixed.
We induct on $m$. The case $m=1$ is equality. For $m\ge2$,
put $\overline Z:=\E_{V_m}Z$. Expanding the definition of entropy and canceling
the terms $\E[\overline Z\log \overline Z]$ gives the chain identity
\[
 \Ent(Z)=\E_{V_{-m}}[\Ent_{V_m}(Z)]+\Ent_{V_{-m}}(\overline Z).
\]
To bound the entropy of the average $\overline Z$, we use convexity of the
entropy functional on a fixed probability space. For bounded
nonnegative functions $h_0,h_1$ and $\delta>0$, put
$h_s:=(1-s)h_0+sh_1+\delta$ for $0<s<1$.
Direct differentiation and Cauchy--Schwarz give
\[
 \frac{d^2}{ds^2}\Ent(h_s)
 =\E\!\left[\frac{(h_1-h_0)^2}{h_s}\right]
  -\frac{(\E[h_1-h_0])^2}{\E[h_s]}\ge0.
\]
Letting $\delta\downarrow0$ proves convexity for bounded nonnegative
functions.

Independence keeps the law of $V_{-m}$ fixed as we average over
$V_m$. Jensen's inequality for the convex entropy functional therefore
gives
\[
 \Ent_{V_{-m}}(\overline Z)
 \le\E_{V_m}[\Ent_{V_{-m}}(Z)].
\]
For each fixed value of $V_m$, apply the induction hypothesis to $Z$ as a
function of the remaining $m-1$ independent coordinates, and then
average over $V_m$:
\[
 \E_{V_m}[\Ent_{V_{-m}}(Z)]
 \le\sum_{j=1}^{m-1}\E_{V_{-j}}[\Ent_{V_j}(Z)].
\]
Together with the chain identity, this proves the entropy inequality.
\end{proof}

\section{Price Bounds for the Prescribed Seed}
\label{app:lower-analytic}

We prove the three estimates in \cref{lem:lower-analytic-seed}: the
revenue at the chosen seed price is greater than $7/50$, shifted
bundle prices earn at most $c=1/4$, and the separate-price quantity
$t\int_t^\infty\nu(x)\,dx$ is at most $9/80$ for every $t>0$.
The aggregate tail bounds bundle revenue, while the intensity bounds
separate revenue. For the two bounds that must hold at every price,
multiplying by an exponential reduces the problem to
showing that a function is positive; a derivative bound then gives
a positive quadratic lower bound.

We first verify the identity that realizes the prescribed aggregate
as a sum of independent items.

For $0<a<b$ and $s\ge0$,
\[
 \int_0^\infty(1-e^{-sx})\frac{e^{-ax}-e^{-bx}}x\,dx
 =\log(1+s/a)-\log(1+s/b).
\]
To see this, write $(e^{-ax}-e^{-bx})/x=\int_a^b e^{-ux}\,du$
and interchange the nonnegative integrals. Taking the two pairs
$(a,b)=(1,4/3)$ and $(4,48)$ gives the Laplace transform in
\cref{sec:lower-analytic}, as well as
$\int_0^\infty\nu(x)\,dx=\log16$.
For $q\in(0,1)$, a random variable that is zero with probability $1-q$ and otherwise exponentially distributed with rate $a$ has transform
\[
 1-q+q\frac{a}{a+s}
 =\frac{1+s/b}{1+s/a},\qquad b:=\frac{a}{1-q}.
\]
Thus the two pairs of rates correspond to the two exponential
components that are present with the prescribed probabilities.
Their intensities are nonnegative and can be split among
independent items. The first moment follows by direct integration:
$\int_0^\infty x\nu(x)\,dx=1-3/4+1/4-1/48=23/48$.

For the revenue at the chosen seed price, it is enough to bound
$e^{-1/2}$ from below. We then verify the two global price bounds by
the derivative argument described above.

\begin{proof}[Proof of \cref{lem:lower-analytic-seed}]
\emph{Revenue at the seed price.}
The aggregate tail gives
\[
 \REV_{\mathrm{seed}}
 =\frac{47e^{-1/2}+88e^{-2}}{288}.
\]
The exponential series implies
\[
 e^{1/2}<1+\frac12+\frac18+\frac{1/48}{1-1/8}
 =\frac{13}{8}+\frac1{42}<\frac{33}{20}.
\]
Consequently,
\[
 \REV_{\mathrm{seed}}
 >\frac{47(20/33)+88(20/33)^4}{288}>\frac7{50}.
\]

\par\smallskip\noindent\emph{Shifted bundle prices.}
Price $c=1/4$ sells surely, and every smaller price earns at most $c$.
For a price $c+t$ with $t>0$, the desired inequality
$(c+t)Q(t)\le c$ is equivalent to $\psi(t)\ge0$, where
\[
 \psi(t):=144e^{4t}-(1+4t)(47e^{3t}+88).
\]
We have $\psi(0)=9$, $\psi'(0)=-105$, and
\[
 \psi''(t)=e^{3t}(2304e^t-1551-1692t).
\]
Using $e^t\ge1+t$ gives
\[
 \psi''(t)\ge e^{3t}(753+612t)\ge750\qquad(t\ge0).
\]
Integrating twice gives
\[
 \psi(t)\ge9-105t+375t^2
 =375\left(t-\frac7{50}\right)^2+\frac{33}{20}>0.
\]
This proves $\REV(c+S_{\mathrm{seed}})=c$.

\par\smallskip\noindent\emph{Separate prices.}
Drop the last negative term of $\nu$ and change variables to obtain
\[
 t\int_t^\infty\nu(x)\,dx
 \le t\int_1^{4/3}\frac{e^{-ut}}u\,du
     +t\int_{4t}^\infty\frac{e^{-x}}x\,dx.
\]
For the first integral, convexity of $u\mapsto e^{-ut}/u$ gives
the trapezoidal bound
\[
 \int_1^{4/3}\frac{e^{-ut}}u\,du
 \le\frac16e^{-t}+\frac18e^{-4t/3}.
\]
For the second integral, we use the estimate
\[
 \int_a^\infty\frac{e^{-x}}x\,dx
 \le\frac{4a+9}{21a}e^{-a}\qquad(a>0).
\]
Indeed, differentiating the proposed upper bound gives
\[
 -\frac{d}{da}\left(\frac{4a+9}{21a}e^{-a}\right)
 -\frac{e^{-a}}a
 =\frac{(2a-3)^2}{21a^2}e^{-a}\ge0,
\]
and integration from $a$ to infinity proves the bound.
Combining the two estimates yields
\[
 t\int_t^\infty\nu(x)\,dx\le \eta(t):=
 \frac t6e^{-t}+\frac t8e^{-4t/3}
 +\left(\frac{4t}{21}+\frac3{28}\right)e^{-4t}.
\]
It remains to show that $\eta(t)\le9/80$ at every price. Put
$\chi(t):=e^{4t}(9/80-\eta(t))$. Direct differentiation gives
\[
 \chi(0)=\frac3{560},\qquad \chi'(0)=-\frac9{280},\qquad
 \chi''(0)=\frac2{15},
\]
and
\[
 e^{-3t}\chi'''(t)=\frac{36}{5}e^t-\frac92(t+1)
   -\frac8{27}(8t+9)e^{-t/3}.
\]
Using $e^t\ge1+t$ and $e^{-t/3}\le1$ for $t\ge0$, we obtain
\[
 e^{-3t}\chi'''(t)\ge\frac{36}{5}(1+t)-\frac92(t+1)
   -\frac8{27}(8t+9)=\frac1{30}+\frac{89}{270}t>0.
\]
Thus $\chi''(t)\ge2/15$ for $t\ge0$, and integrating twice gives
\[
 \chi(t)\ge\frac3{560}-\frac9{280}t+\frac1{15}t^2
 \ge\frac{36(t-1/4)^2+3/4}{560}>0.
\]
The second inequality uses $1/15\ge36/560$.
Hence $\eta(t)<9/80$ for all $t\ge0$, proving the separate-price bound.
\end{proof}

\begingroup
\setlength{\emergencystretch}{1em}
\printbibliography
\endgroup

\end{document}

%% file: preamble.tex
\newtheorem{theorem}{Theorem}[section]
\newtheorem{corollary}[theorem]{Corollary}
\newtheorem{lemma}[theorem]{Lemma}

\crefname{claim}{Claim}{Claims}
\Crefname{claim}{Claim}{Claims}

\newenvironment{proof-sketch}{\noindent{\bf Sketch of Proof}\hspace*{1em}}{\qed\bigskip}
\newenvironment{proof-idea}{\noindent{\bf Proof Idea}\hspace*{1em}}{\qed\bigskip}
\newenvironment{proof-of-lemma}[1]{\noindent{\bf Proof of Lemma #1}\hspace*{1em}}{\qed\bigskip}
\newenvironment{proof-attempt}{\noindent{\bf Proof Attempt}\hspace*{1em}}{\qed\bigskip}

\makeatletter
\def\fnum@figure{{\bf Figure \thefigure}}
\def\fnum@table{{\bf Table \thetable}}
\long\def\@mycaption#1[#2]#3{\addcontentsline{\csname
  ext@#1\endcsname}{#1}{\protect\numberline{\csname
  the#1\endcsname}{\ignorespaces #2}}\par
  \begingroup
    \@parboxrestore
    \small
    \@makecaption{\csname fnum@#1\endcsname}{\ignorespaces #3}\par
  \endgroup}
\def\mycaption{\refstepcounter\@captype \@dblarg{\@mycaption\@captype}}
\makeatother

\newcommand{\mathify}[1]{\ifmmode{#1}\else\mbox{$#1$}\fi}
\newcommand{\bigO}O
\newcommand{\set}[1]{\mathify{\left\{ #1 \right\}}}

\newcommand{\remove}[1]{}
\newcommand{\ignore}[1]{}



%% file: references.bib
@article{BabaioffEtAl2020,
  author  = {Babaioff, Moshe and Immorlica, Nicole and Lucier, Brendan and Weinberg, S. Matthew},
  title   = {A Simple and Approximately Optimal Mechanism for an Additive Buyer},
  journal = {Journal of the ACM},
  year    = {2020},
  volume  = {67},
  number  = {4},
  pages   = {24:1--24:40},
  doi     = {10.1145/3398745},
  url     = {https://doi.org/10.1145/3398745},
  note    = {Preliminary version in FOCS 2014}
}

@misc{CaiEtAl2026,
  ids       = {cai2026improved},
  author    = {Cai, Yang and Gupta, Vineet and Jiang, Yanchen and Liaw, Christopher and Mehta, Aranyak and Velegkas, Grigoris and Wang, Di},
  title     = {Improved Revenue Guarantees for Selling Separately and Bundling},
  year      = {2026},
  eprinttype = {arxiv},
  eprint    = {2609.28873},
  url       = {https://arxiv.org/abs/2609.28873}
}

@misc{GiannakopoulosHahn2026,
  ids       = {Giannakopoulos2026},
  author    = {Giannakopoulos, Yiannis and Hahn, Johannes},
  title     = {On the Power of Determinism in Multi-Item Auctions},
  year      = {2026},
  eprinttype = {arxiv},
  eprint    = {2609.35711},
  url       = {https://arxiv.org/abs/2609.35711v2},
  note      = {Version 2, October 1, 2026}
}

@inproceedings{BeyhaghiWeinberg2019,
  author    = {Beyhaghi, Hedyeh and Weinberg, S. Matthew},
  title     = {Optimal (and Benchmark-Optimal) Competition Complexity
               for Additive Buyers over Independent Items},
  booktitle = {Proceedings of the 51st Annual ACM SIGACT Symposium
               on Theory of Computing},
  series    = {STOC '19},
  year      = {2019},
  pages     = {686--696},
  publisher = {Association for Computing Machinery},
  address   = {New York, NY, USA},
  doi       = {10.1145/3313276.3316405},
  url       = {https://doi.org/10.1145/3313276.3316405}
}

@inproceedings{BeyhaghiEtAl2025,
  ids       = {beyhaghi2025competition},
  author    = {Beyhaghi, Hedyeh and Cai, Linda and Feng, Yiding and Li, Yingkai and Weinberg, S. Matthew},
  title     = {Competition Complexity in Multi-Item Auctions: Beyond {VCG} and Regularity},
  booktitle = {Proceedings of the 26th ACM Conference on Economics and Computation},
  series    = {EC '25},
  year      = {2025},
  pages     = {89},
  publisher = {Association for Computing Machinery},
  doi       = {10.1145/3736252.3742498},
  url       = {https://arxiv.org/abs/2506.09291}
}

@inproceedings{MaSimchiLevi2021,
  ids       = {ma2021reaping},
  author    = {Ma, Will and Simchi-Levi, David},
  title     = {Reaping the Benefits of Bundling under High Production Costs},
  booktitle = {Proceedings of the 24th International Conference on Artificial Intelligence and Statistics},
  year      = {2021},
  volume    = {130},
  pages     = {1342--1350},
  editor    = {Banerjee, Arindam and Fukumizu, Kenji},
  series    = {Proceedings of Machine Learning Research},
  publisher = {PMLR},
  url       = {https://proceedings.mlr.press/v130/ma21b.html}
}

@article{BulowKlemperer1996,
  author  = {Bulow, Jeremy and Klemperer, Paul},
  title   = {Auctions versus Negotiations},
  journal = {American Economic Review},
  year    = {1996},
  volume  = {86},
  number  = {1},
  pages   = {180--194},
  month   = mar,
  url     = {https://www.jstor.org/stable/2118262}
}

@inproceedings{CaiDevanurWeinberg2016,
  author    = {Cai, Yang and Devanur, Nikhil R. and Weinberg, S. Matthew},
  title     = {A Duality-Based Unified Approach to {Bayesian}
               Mechanism Design},
  booktitle = {Proceedings of the 48th Annual ACM SIGACT Symposium
               on Theory of Computing},
  series    = {STOC '16},
  year      = {2016},
  pages     = {926--939},
  publisher = {Association for Computing Machinery},
  address   = {New York, NY, USA},
  doi       = {10.1145/2897518.2897645},
  url       = {https://doi.org/10.1145/2897518.2897645}
}

@inproceedings{DerakhshanEtAl2024,
  author    = {Derakhshan, Mahsa and Ryu, Emily and
               Weinberg, S. Matthew and Xue, Eric},
  title     = {Settling the Competition Complexity of Additive Buyers
               over Independent Items},
  booktitle = {Proceedings of the 25th ACM Conference on Economics
               and Computation},
  series    = {EC '24},
  year      = {2024},
  pages     = {420--446},
  publisher = {Association for Computing Machinery},
  address   = {New York, NY, USA},
  doi       = {10.1145/3670865.3673627},
  url       = {https://doi.org/10.1145/3670865.3673627}
}

@inproceedings{EdenEtAl2017,
  author    = {Eden, Alon and Feldman, Michal and Friedler, Ophir and
               Talgam-Cohen, Inbal and Weinberg, S. Matthew},
  title     = {The Competition Complexity of Auctions:
               A {Bulow--Klemperer} Result for Multi-Dimensional Bidders},
  booktitle = {Proceedings of the 2017 ACM Conference on Economics
               and Computation},
  series    = {EC '17},
  year      = {2017},
  pages     = {343--343},
  publisher = {Association for Computing Machinery},
  address   = {New York, NY, USA},
  doi       = {10.1145/3033274.3085115},
  url       = {https://doi.org/10.1145/3033274.3085115}
}

@book{BoucheronLugosiMassart2013,
  author    = {St{\'e}phane Boucheron and G{\'a}bor Lugosi
               and Pascal Massart},
  title     = {Concentration Inequalities:
               A Nonasymptotic Theory of Independence},
  publisher = {Oxford University Press},
  year      = {2013},
  doi       = {10.1093/acprof:oso/9780199535255.001.0001}
}

@article{Ledoux1999,
  author  = {Michel Ledoux},
  title   = {Concentration of Measure and Logarithmic
             {Sobolev} Inequalities},
  journal = {S{\'e}minaire de Probabilit{\'e}s},
  volume  = {33},
  pages   = {120--216},
  year    = {1999}
}

@article{Chafai2004,
  author  = {Djalil Chafa{\"i}},
  title   = {Entropies, Convexity, and Functional Inequalities:
             On {$\Phi$}-Entropies and {$\Phi$}-{Sobolev} Inequalities},
  journal = {Journal of Mathematics of Kyoto University},
  volume  = {44},
  number  = {2},
  pages   = {325--363},
  year    = {2004},
  doi     = {10.1215/kjm/1250283556}
}

@article{DiCrescenzo2009,
title = {On cumulative entropies},
journal = {Journal of Statistical Planning and Inference},
volume = {139},
number = {12},
pages = {4072-4087},
year = {2009},
issn = {0378-3758},
doi = {https://doi.org/10.1016/j.jspi.2009.05.038},
url = {https://www.sciencedirect.com/science/article/pii/S0378375809001621},
author = {Antonio {Di Crescenzo} and Maria Longobardi}
}

@inproceedings{FeldmanFriedlerRubinstein2018,
  author    = {Feldman, Michal and Friedler, Ophir and Rubinstein, Aviad},
  title     = {{99\%} Revenue via Enhanced Competition},
  booktitle = {Proceedings of the 2018 ACM Conference on Economics
               and Computation},
  series    = {EC '18},
  year      = {2018},
  pages     = {443--460},
  publisher = {Association for Computing Machinery},
  address   = {New York, NY, USA},
  doi       = {10.1145/3219166.3219202},
  url       = {https://doi.org/10.1145/3219166.3219202}
}

@article{HartNisan2017,
  author  = {Hart, Sergiu and Nisan, Noam},
  title   = {Approximate Revenue Maximization with Multiple Items},
  journal = {Journal of Economic Theory},
  year    = {2017},
  volume  = {172},
  pages   = {313--347},
  doi     = {10.1016/j.jet.2017.09.001},
  url     = {https://doi.org/10.1016/j.jet.2017.09.001}
}

@article{Myerson1981,
  author  = {Myerson, Roger B.},
  title   = {Optimal Auction Design},
  journal = {Mathematics of Operations Research},
  year    = {1981},
  volume  = {6},
  number  = {1},
  pages   = {58--73},
  doi     = {10.1287/moor.6.1.58},
  url     = {https://doi.org/10.1287/moor.6.1.58}
}

@inproceedings{Rubinstein2016,
  author       = {Rubinstein, Aviad},
  title        = {On the Computational Complexity of Optimal Simple Mechanisms},
  booktitle    = {Proceedings of the 2016 ACM Conference on Innovations in Theoretical Computer Science},
  series       = {ITCS '16},
  year         = {2016},
  pages        = {21--28},
  publisher    = {Association for Computing Machinery},
  doi          = {10.1145/2840728.2840736},
}

@inproceedings{DaskalakisDeckelbaumTzamos2013,
  author    = {Daskalakis, Constantinos and Deckelbaum, Alan and Tzamos, Christos},
  title     = {Mechanism Design via Optimal Transport},
  booktitle = {Proceedings of the Fourteenth ACM Conference on Electronic Commerce},
  series    = {EC '13},
  year      = {2013},
  pages     = {269--286},
  publisher = {Association for Computing Machinery},
  doi       = {10.1145/2482540.2482593},
  url       = {https://doi.org/10.1145/2482540.2482593},
}

@inproceedings{DaskalakisDeckelbaumTzamos2014,
  author    = {Daskalakis, Constantinos and Deckelbaum, Alan and Tzamos, Christos},
  title     = {The Complexity of Optimal Mechanism Design},
  booktitle = {Proceedings of the Twenty-Fifth Annual ACM-SIAM Symposium on Discrete Algorithms},
  series    = {SODA '14},
  year      = {2014},
  pages     = {1302--1318},
  publisher = {Society for Industrial and Applied Mathematics},
  doi       = {10.1137/1.9781611973402.96},
  url       = {https://doi.org/10.1137/1.9781611973402.96},
}

@article{DaskalakisDeckelbaumTzamos2017,
  author    = {Daskalakis, Constantinos and Deckelbaum, Alan and Tzamos, Christos},
  title     = {Strong Duality for a Multiple-Good Monopolist},
  journal   = {Econometrica},
  year      = {2017},
  volume    = {85},
  number    = {3},
  pages     = {735--767},
  doi       = {10.3982/ECTA12618},
  url       = {https://doi.org/10.3982/ECTA12618},
  note      = {Preliminary version in EC 2015}
}

@article{Daskalakis2015,
  author  = {Daskalakis, Constantinos},
  title   = {Multi-Item Auctions Defying Intuition?},
  journal = {ACM SIGecom Exchanges},
  year    = {2015},
  volume  = {14},
  number  = {1},
  pages   = {41--75},
  doi     = {10.1145/2845926.2845928},
  url     = {https://www.sigecom.hosting.acm.org/exchanges/volume_14/1/DASKALAKIS.pdf}
}

@article{LiYao2013,
  author  = {Li, Xinye and Yao, Andrew Chi-Chih},
  title   = {On Revenue Maximization for Selling Multiple Independently Distributed Items},
  journal = {Proceedings of the National Academy of Sciences},
  year    = {2013},
  volume  = {110},
  number  = {28},
  pages   = {11232--11237},
  doi     = {10.1073/pnas.1309533110},
  url     = {https://doi.org/10.1073/pnas.1309533110}
}

@misc{Kupfer2016,
  author    = {Kupfer, Ron},
  title     = {A Note on the Ratio of Revenues Between Selling in a Bundle and Separately},
  year      = {2016},
  eprinttype = {arxiv},
  eprint    = {1611.09613},
  url       = {https://arxiv.org/abs/1611.09613}
}

@article{ManelliVincent2006,
  author  = {Manelli, Alejandro M. and Vincent, Daniel R.},
  title   = {Bundling as an Optimal Selling Mechanism for a Multiple-Good Monopolist},
  journal = {Journal of Economic Theory},
  year    = {2006},
  volume  = {127},
  number  = {1},
  pages   = {1--35},
  doi     = {10.1016/j.jet.2005.08.007},
  url     = {https://doi.org/10.1016/j.jet.2005.08.007}
}

@article{Rochet1987,
  author  = {Rochet, Jean-Charles},
  title   = {A Necessary and Sufficient Condition for Rationalizability in a Quasi-Linear Context},
  journal = {Journal of Mathematical Economics},
  year    = {1987},
  volume  = {16},
  number  = {2},
  pages   = {191--200},
  doi     = {10.1016/0304-4068(87)90007-3},
  url     = {https://doi.org/10.1016/0304-4068(87)90007-3}
}

@inproceedings{RoughgardenTalgamCohenYan2012,
  author    = {Roughgarden, Tim and Talgam-Cohen, Inbal and Yan, Qiqi},
  title     = {Supply-Limiting Mechanisms},
  booktitle = {Proceedings of the 13th ACM Conference on Electronic Commerce},
  series    = {EC '12},
  year      = {2012},
  pages     = {844--861},
  publisher = {Association for Computing Machinery},
  doi       = {10.1145/2229012.2229077},
  url       = {https://doi.org/10.1145/2229012.2229077}
}
